\documentclass{article}

\usepackage{amsmath,amsthm} 
\usepackage{amssymb,mathrsfs} 
\usepackage{hyperref}
\usepackage{graphicx} 
\usepackage{color,subfigure} 
\usepackage{enumerate}  
\usepackage{fullpage}
\usepackage{cancel,ulem} 

\newtheorem{theorem}{Theorem}
\newtheorem{lemma}{Lemma}

\newtheorem{proposition}{Proposition}

\DeclareMathAlphabet{\mathpzc}{OT1}{pzc}{m}{it}
\newcommand{\dps}{\displaystyle } 
\newcommand{\rme}{\mathrm{e}} 
\newcommand{\ri}{\mathrm{i}}

\newcommand{\R}{\mathbb{R}}
\newcommand{\Id}{\mathrm{Id}}

\newcommand{\cL}{\mathcal{L}}

\renewcommand{\leq}{\leqslant}
\renewcommand{\geq}{\geqslant}
\renewcommand{\le}{\leqslant}

\newcommand{\D}{d}
\newcommand{\dd}{d}
\newcommand{\cD}{\mathcal{D}}
\newcommand{\cE}{\mathcal{E}}
\newcommand{\cLxt}{\mathcal{L}_{\xi,\tau}}
\newcommand{\la}{\left\langle}
\newcommand{\ra}{\right\rangle_{L^2(\mu)}}
\newcommand{\raa}{\right\rangle_{a(\xi)}}
\newcommand{\Ltwo}{{L^2(\mu)}}
\newcommand{\sB}{\mathscr{B}}
\newcommand{\mx}{\min(\xi,\xi^{-1})}
\newcommand{\Li}{\mathcal{K}}
\newcommand{\cLham}{\mathcal{L}_{\rm ham}}
\newcommand{\cLFD}{\mathcal{L}_{\rm FD}}
\newcommand{\cLpert}{\mathcal{L}_{\rm pert}}
\newcommand{\cLovd}{\mathcal{L}_{\rm ovd}}
\newcommand{\sH}{\mathscr{H}}

\begin{document}

  
\title{Convergence rates for nonequilibrium Langevin dynamics}
\author{A. Iacobucci$^{1}$, S. Olla$^{1}$ and G. Stoltz$^{2}$ \\
{\small $^{1}$ CEREMADE, UMR-CNRS, Universit\'e Paris Dauphine, PSL} \\
{\small Place du Mar\'echal De Lattre De Tassigny, 75775 Paris Cedex 16, France} \\
{\small $^{2}$ Universit\'e Paris-Est, CERMICS (ENPC), Inria, F-77455 Marne-la-Vall\'ee, France} \\
}

\date{\today}

\maketitle

\abstract{
  We study the exponential convergence to the stationary state for nonequilibrium Langevin dynamics, by a perturbative approach based on hypocoercive techniques developed for equilibrium Langevin dynamics. The Hamiltonian and overdamped limits (corresponding respectively to frictions going to zero or infinity) are carefully investigated. In particular, the maximal magnitude of admissible perturbations are quantified as a function of the friction. Numerical results based on a Galerkin discretization of the generator of the dynamics confirm the theoretical lower bounds on the spectral gap.
}

\bigskip



\section{Introduction}

Langevin dynamics are a commonly used model to describe the evolution of systems at constant temperature, \textit{i.e.} in contact with a heat bath at equilibrium at a given temperature. The interaction with the heat bath is modeled by a dissipative term and a random fluctuating term related by a fluctuation-dissipation relation. This dynamics is therefore a stochastic perturbation of the Hamiltonian dynamics. It can be used to sample configurations of a physical system at equilibrium, according to the canonical ensemble, which allows to numerically estimate average macroscopic properties; see for instance~\cite{LM15} for a mathematically oriented introduction to molecular dynamics. The Langevin dynamics is by now well-understood, both for its theoretical properties and its discretization, see for instance the review article~\cite{LS16}.

On the other hand, the properties of nonequilibrium Langevin dynamics, as obtained for instance by the addition of a non-gradient drift term, have been less investigated. We consider in this work the following paradigmatic nonequilibrium Langevin dynamics in a compact position space with periodic boundary conditions (see Section~\ref{sec:description} for a more precise description):
\begin{equation}
  \label{eq:Langevin}
  \left\{ \begin{aligned}
    \dd q_t & = \frac{p_t}{m} \D t, \\*
    \dd p_t & =  ( -\nabla U(q_t) + \tau F) \D t - \xi \frac{p_t}{m}  \D t 
    + \sqrt{\frac{2\xi}{\beta}}  \D W_t,
\end{aligned} \right. 
\end{equation}
where the position~$q$ belongs to~$\cD = (2\pi \mathbb{T})^d$ (with $\mathbb{T} = \mathbb{R}/\mathbb{Z}$ the one-dimensional unit torus), the momentum $p$ is in~$\mathbb{R}^d$, $W_t$ is a standard $d$-dimensional Brownian motion, and $U:\cD \to \mathbb{R}$ is the potential energy function. We assume~$U$ to be smooth in all this work, and denote by $\cE = \cD \times \mathbb{R}$ the phase-space of the system. The dynamics~\eqref{eq:Langevin} is parametrized by several constants: the mass~$m>0$ of the particles (for simplicity, we consider a single mass~$m$, although our results can be extended to account for more general mass matrices), the inverse temperature $\beta > 0$, the friction~$\xi>0$ and the magnitude of the external force~$\tau \in \mathbb{R}$, with given direction $F \in \mathbb{S}^{d-1}$ (\textit{i.e.} $F \in \mathbb{R}^d$ with $|F|=1$). Let us emphasize that the external forcing induced by $\tau F$ is indeed non-gradient since $\tau F$ is not the gradient of a smooth, periodic function. Let us also mention that our analysis could be extended to more general non-gradient forcings $F(q)$ genuinely depending on the position variable, as long as the function $F$ is sufficiently smooth.

There are two interesting limiting regimes that can be considered: (i) the limit $\xi \to 0$, which corresponds to a Hamiltonian limit; (ii) the limit $m \to 0$ or $\xi \to +\infty$ (with, in the latter case, a time rescaling by a factor $\xi$), which corresponds to an overdamped limit. More precisely, fixing for instance $m=1$ and setting $\tau = 0$, a simple proof shows that $q_{\xi t}$, the solution of~\eqref{eq:Langevin} observed at time $\xi t$, converges in law to the solution $Q_t$ of the overdamped Langevin dynamics (see for instance~\cite[Proposition~2.14]{LRS10})
\begin{equation}
  \label{eq:overdamped}
  dQ_t = -\nabla U(Q_t) \, dt + \sqrt{\frac2\beta} \, d\widetilde{W}_t.
\end{equation}

In the absence of external forcing ($\tau = 0$), the system described by~\eqref{eq:Langevin} has an {\it equilibrium} stationary measure given explicitly by the canonical Gibbs measure:
\[
\mu(dq\,dp ) = Z_\mu^{-1} \rme^{-\beta H(q,p)} \, dq \, dp, 
\qquad
H(q,p) = U(q) + \frac{p^2}{2m}, 
\qquad
Z_\mu = \int_\cE \rme^{-\beta H}.
\]
When $\tau \neq 0$, there exists a unique {\it nonequilibrium} stationary measure, but it cannot be computed explicitly. 

Although the Langevin dynamics~\eqref{eq:Langevin} is not elliptic (the noise acts only on the momenta and not directly on the positions), it can be shown to be hypoelliptic. The rate of the exponential convergence to the stationary state can be obtained by Lyapunov techniques, see for instance~\cite{MSH02,rey-bellet}. The corresponding convergence rates are however usually not very explicit in terms of the parameters of the dynamics. In particular, it is difficult to make explicit their dependence on~$\xi$. A more quantitative approach is based on estimates in $L^2(\mu)$, where exponential convergence rates for the law of the process towards~$\mu$ can be obtained by a careful use of commutator identities, as pioneered in~\cite{Talay02,EH03,HN04} and later abstracted in the theory of hypocoercivity~\cite{Villani09}. The application of this theory to Langevin dynamics allows to quantify the convergence rates in terms of the parameters of the dynamics; see for instance~\cite{HP08} for the Hamiltonian limit $\xi \to 0$ and~\cite{LMS15,LS16} for partial results on the overdamped limit $\xi \to +\infty$. A more direct route to prove the convergence was subsequently proposed in~\cite{DMS09,DMS15}, which makes it even easier to quantify convergence rates; see~\cite{DKMS13} for a complete study on the dependence of a parameter similar to~$\xi$ by this approach for a dynamics similar to an equilibrium Langevin dynamics, as well as~\cite{AAS15} for sharp estimates for equilibrium Langevin dynamics and a harmonic potential energy function. An extension to nonlinear potentials was also provided in the latter work, based on a non-symmetric Bakry--Emery condition; although further work is required to use this approach in our context because the results are stated for equilibrium dynamics with potentials~$V$ which are convex and such that $\sqrt{\max V''} - \sqrt{\min V''} \leq \xi$. Let us finally mention the recent work~\cite{EGZ17} which provides the qualitative rates of convergence for Langevin dynamics by a coupling strategy.

Our aim in this work is to investigate the exponential convergence to the stationary state for nonequilibrium Langevin dynamics. The new contributions of this work are the following:
\begin{enumerate}[(1)]
\item we provide some technical variations/improvements on the theoretical side by
\begin{enumerate}[(a)] 
\item giving explicit decay estimates for nonequilibrium Langevin dynamics, both in the Hamiltonian limit $\xi\to0$ and in the overdamped limit $\xi\to+\infty$, thereby extending the results obtained at equilibrium. Similar convergence results were very recently stated for a different nonequilibrium model and a fixed friction in~\cite{BHM16};
\item deriving hypocoercive estimates in a degenerate $H^1(\mu)$ norm, which allows to obtain lower bounds on convergence rates which are more explicit than the ones obtained with a non-degenerate $H^1(\mu)$ norm. Such degenerate norms where already considered in~\cite{Talay02}, and were also recently used in~\cite{Baudoin13,OL15,Baudoin16};
\item comparing the results obtained by either the standard hypocoercive approach followed by hypoelliptic regularization, or the direct $L^2$ approach of~\cite{DMS09,DMS15}.
\end{enumerate}
\item on the numerical side, we perform a study of the spectral gap of the generator of the dynamics (which is related to the exponential convergence of the law of the process to the invariant measure), as a function of the friction~$\xi$ and the magnitude~$\tau$ of the external forcing, in order to assess the sharpness of the bounds provided by the theoretical results. Let us emphasize that we consider a situation where the invariant measure does not have an explicit expression, so that it is difficult to rely on Monte-Carlo techniques to estimate the convergence rate as in~\cite{DKMS13}. We consider instead a Galerkin discretization of the generator of the dynamics. 
\end{enumerate}

\medskip

This article is organized as follows. We first start by describing the stationary state of the nonequilibrium Langevin dynamics we consider in Section~\ref{sec:description}. We next state in Section~\ref{sec:exp_cv} the exponential convergence to~0 of the evolution semigroup in various functional settings. Our emphasis is on carefully estimating the scaling of these rates with respect to the friction~$\xi$ and the magnitude~$\tau$ of the external forcing. The estimates we obtain provide lower bounds on the spectral gap of the Fokker--Planck operator. The relevance/sharpness of these bounds is investigated from a numerical viewpoint in Section~\ref{sec:num} by a Galerkin discretization. The longest proofs of the results presented in Section~\ref{sec:exp_cv} are postponed to Section~\ref{sec:proofs}.

\section{Stationary states of nonequilibrium Langevin dynamics}
\label{sec:description}

We make precise in this section some properties of the nonequilibrium Langevin dynamics. We first state a result on the existence and uniqueness of the steady-state, as well as (non explicit) exponential convergence results (see Section~\ref{sec:ex_uniq_inv_meas}). We next give some simple properties of the stationary state in Section~\ref{sec:qualitative_ppties}. We conclude with a perturbative expansion of the steady-state in terms of the magnitude~$\tau$ of the external forcing, making explicit the admissible magnitude of the perturbation~$\tau$ in terms of the friction~$\xi$ (see Section~\ref{sec:pert_meas}).

\subsection{Existence and uniqueness of the invariant measure}
\label{sec:ex_uniq_inv_meas}

We recall in this section a result on the existence and uniqueness of the invariant measure for any value of~$\tau$, as well as exponential convergence rates in weighted $L^\infty$ spaces (see~\cite[Proposition~5.2]{LS16}). To state this result, we denote by $\cLxt$ the generator of the dynamics~\eqref{eq:Langevin}:
\[
\cLxt = \cLham + \xi \cLFD + \tau \cLpert,
\]
where the generator associated with the Hamiltonian part of the dynamics, the fluctuation/dissipation and the external perturbation respectively read
\[
\cLham = \frac{p}{m} \cdot \nabla_q - \nabla U(q) \cdot \nabla_p, 
\qquad 
\cLFD = -\frac{p}{m}\cdot \nabla_p + \frac1\beta \Delta_p, 
\qquad
\cLpert = F\cdot \nabla_p.
\]
We also introduce the Lyapunov functions $\Li_n(q,p) = 1 + |p|^n$ for $n\geq 2$, and the functional spaces
\[
L^\infty_{\Li_n} = \left\{ \varphi \textrm{ measurable}, \ \| \varphi \|_{L^\infty_{\Li_n}} = \left\|\frac{\varphi}{\Li_n}\right\|_{L^\infty} < +\infty \right\}.
\]

\begin{proposition}
\label{prop:unif_cv_cL_eta}
Fix $\tau_* > 0$ and $\xi > 0$. For any $\tau \in [-\tau_*,\tau_*]$, the dynamics~\eqref{eq:Langevin} admits a unique invariant probability measure which admits a $C^\infty$ density $\psi_\tau(q,p)$ with respect to the Lebesgue measure. Moreover, for any $n \geq 2$, there exist $C_n,\lambda_n>0$ (depending on $\tau_*$) such that, for any $\tau \in [-\tau_*,\tau_*]$ and for any $\varphi \in L^\infty_{\Li_n}(\cE)$,
\[
\forall t \geq 0, \qquad \left\| \rme^{t \cL_\tau}\varphi - \int_\cE \varphi \,\psi_\tau \right\|_{L^\infty_{\Li_n}} \leq C_n \rme^{-\lambda_n t} \left\| \varphi \right\|_{L^\infty_{\Li_n}}.
\] 
\end{proposition}

Let us however emphasize that it is difficult to quantify the above convergence rates in terms of the parameters of the dynamics (such as the friction $\xi$ or the potential~$U$) and the magnitude of the external forcing~$\tau$, since the proof relies on minorization conditions for which the dependence on the parameters is not very explicit. 

In order to write more explicit convergence results, it will be convenient to work in (subspaces) of $L^2(\mu)$, see Section~\ref{sec:exp_cv}. In the sequel, we consider by default all operators as defined (by their extensions) on the Hilbert space $L^2(\mu)$, the adjoint~$A^*$ of a (closable) operator~$A$ being defined with respect to the associated canonical scalar product. Working in $L^2(\mu)$ however requires that the invariant measure itself admits a density $h_\tau$ with respect to~$\mu$ which is in $L^2(\mu)$, \textit{i.e.} 
\begin{equation}
\label{eq:psi_tau_h_tau}
\psi_\tau = h_\tau \mu, \qquad h_\tau \in \Ltwo. 
\end{equation}
Such a result is provided by perturbation results for linear operators, see Section~\ref{sec:pert_meas}. In this setting, the invariance of the measure $\psi_\tau$ can be translated into the following Fokker-Planck equation:
\begin{equation}
  \label{eq:FP_L2}
  \cLxt^* h_\tau = 0, \qquad \int_\cE h_\tau \, d\mu = 1,
\end{equation}
with $h_\tau \geq 0$. More explicitly, the action of $\cLxt^*$, the adjoint of $\cLxt$ considered as an operator on~$L^2(\mu)$, is given by
\begin{equation}
  \label{eq:16}
  \cLxt^* = - \cLham + \xi \cLFD - \tau F\cdot \nabla_p + \frac{\tau\beta}{m} F\cdot p.
\end{equation}

\subsection{Qualitative properties of the steady-state}
\label{sec:qualitative_ppties}

We provide here some simple properties of the stationary state. The first one is that, due to the constant force applied on the momenta, a non-zero velocity builds up in the system. More precisely, denoting by $\mathbb{E}_\tau$ the expectation with respect to~$\psi_\tau$, this average velocity~$v_\tau$ satisfies
\begin{equation}
  \label{eq:v_tau}
  v_\tau = \mathbb{E}_\tau\left(\frac{p}{m}\right) = \frac1\xi\Big( \tau F - \mathbb{E}_\tau(\nabla U) \Big).
\end{equation}
This relation follows from the following identities (the first equality arising from the invariance of $\psi_\tau$ through the nonequilibrium Langevin dynamics~\eqref{eq:Langevin}):
\[
\int_\cE \left(\cLxt p \right) \psi_\tau = 0 = - \xi \int_\cE \frac{p}{m} \, \psi_\tau - \int_\cE \nabla U \, \psi_\tau + \tau F.
\]

Another relation is obtained by computing the average of $\cLxt \mathcal H$ with respect to~$\psi_\tau$, where $\mathcal{H}$ is the Hamiltonian of the system:
\[
\mathcal{H}(q,p) = \frac{p^2}{2m} + U(q).
\] 
Since 
\[
\cLxt \mathcal{H}(q,p) = \xi \left(-\frac{p^2}{m^2} + \frac{1}{\beta m}\right) + \tau F \cdot \frac{p}{m},
\]
it follows that
\[
\tau F\cdot v_\tau = - \frac{\xi}{m} \left(\frac1\beta - \mathbb{E}_\tau\left[\frac{p^2}{m}\right] \right). 
\]
The above identity expresses the energy conservation: the work performed by the non-conservative force $\tau F$ is equal to the heat, \textit{i.e.} the energy flow into the heat bath. 

Let us now turn to the entropy production rate of the stationary state, which is proportional to $\tau F \cdot v_\tau$ by the above discussion. By energy conservation of the system and its environment, the heat flow from the thermostat equals the instantaneous work performed by the driving force. We therefore study the latter quantity. We consider to this end $\cLxt (\ln h_\tau)$ (where $h_\tau$ is defined in~\eqref{eq:psi_tau_h_tau}). The following computations are formal but could be made rigorous upon obtaining appropriate estimates on $\nabla_p h_\tau$. Since
\[
h_\tau \cLxt \left( \ln h_\tau \right) = \cLxt h_\tau - \frac {\xi}{\beta} \frac{|\nabla_p h_\tau|^2}{h_\tau},
\]
and
\[
\int_\cE \left(\cLxt h_\tau\right) d\mu = \int_\cE \left(\cLxt^*\mathbf{1}\right) h_\tau \, d\mu = \tau\beta F \cdot \int_\cE \frac{p}{m} h_\tau \, d\mu = {\tau\beta} F \cdot v_\tau,
\]
it follows, after integration with respect to~$\mu$ and in view of~\eqref{eq:FP_L2}, 
\[
\int_\cE h_\tau \cLxt \left( \ln h_\tau \right) d\mu = \int_\cE \left(\cLxt^*h_\tau\right) \ln h_\tau \, d\mu= 0 = {\tau\beta} F \cdot v_\tau - \frac{\xi}{\beta} \int \frac{|\nabla_p h_\tau|^2}{h_\tau} d\mu.
\]
This shows that the entropy production rate is positive. 
Using~\eqref{eq:v_tau}, we also obtain the following bound on a degenerate Fisher information:
\[
\int_\cE \frac{|\nabla_p h_\tau|^2}{h_\tau} \, d\mu = \frac{\beta^2 \tau}{\xi} F \cdot v_\tau
= \left(\frac{\beta \tau}{\xi}\right)^2 - \frac{\beta^2\tau}{\xi^2} F\cdot  \mathbb{E}_\tau(\nabla U)
\le \left(\frac{\beta \tau}{\xi}\right)^2+ \frac{\beta^2\tau}{\xi^2} \|\nabla U\|_\infty,
\]
which provides some information on the distance of the stationary distribution of the momenta to the equilibrium Gaussian distribution with variance $m/\beta$.

A last information on the stationary state is provided by the Einstein relation, first proven in~\cite{Rodenhausen89}, which provides a linear response result on the average velocity~$v_\tau$:
\begin{equation}
  \label{eq:ER}
  \frac d{d\tau} v_\tau \Big|_{\tau = 0} = \beta DF,
\end{equation}
where $D \in \mathbb{R}^{d \times d}$ is the diffusivity of the system at $\tau = 0$ \textit{i.e.} the asymptotic variance:
\begin{equation}
  \label{eq:diffusivity}
  D := \lim_{t\to\infty} \frac{1}{2t} \mathbb{E}_0\left[ \left(\int_0^t \frac{p_s}{m} \, ds\right) \otimes \left(\int_0^t \frac{p_s}{m} \, ds\right) \right],
\end{equation}
with $\mathbb{E}_0$ the expectation over all initial conditions~$(q_0,p_0) \sim \mu$ and all realizations of the Brownian motion~$W_t$. This relation also follows from the first order term in the expansion of $h_\tau$ recalled in Section~\ref{sec:pert_meas}.

\subsection{Perturbative expansion of the invariant measure}
\label{sec:pert_meas}

We prove in this section that the decomposition~\eqref{eq:psi_tau_h_tau} is well defined, and actually give a complete characterization of $h_\tau$ as some series expansion for $\tau$ sufficiently small. Define the subspace of $L^2(\mu)$ composed of functions with average~0 with respect to~$\mu$:
\[
L^2_0(\mu) = \left\{ f \in L^2(\mu), \ \int_\cE f \, d\mu = 0 \right\}.
\]
We also denote by $\sB(X)$ the operator norm on a Banach space~$X$.

The operator $\cL_{\xi,0}$ is invertible on $L^2_0(\mu)$. This can be seen as a particular case of the convergence results provided in Section~\ref{sec:exp_cv} (for $\tau = 0$). In fact, it can be proved that there exists $K > 0$ such that
\[
\forall \xi \in (0,+\infty), \qquad \left\| \cL_{\xi,0}^{-1} \right\|_{\sB(L^2_0(\mu))} \leq \frac{K}{\mx}.
\]
see also~\cite{HP08,LMS15,LS16}. In addition, recalling that $|F|=1$, we obtain, for a smooth and compactly supported function $\varphi \in L^2_0(\mu)$,
\[
\left\| \cL_{\rm pert} \varphi \right\|^2_{L^2(\psi_0)} \leq \| \nabla_p \varphi \|_\Ltwo^2 = -\frac{\beta}{\xi} \la \varphi, \cL_{\xi,0}\varphi \ra \leq \frac{\beta}{\xi} \| \cL_{\xi,0} \varphi \|_{L^2(\psi_0)} \| \varphi \|_{L^2(\psi_0)}.
\]
We next project $\cL_{\rm pert} \varphi$ onto $L^2_0(\mu)$, replace $\varphi$ by $\cL_{\xi,0}^{-1} \phi$ and take the supremum over $\phi \in L^2_0(\mu)$ to obtain a bound on the operator $\cL_{\rm pert} \cL_{\xi,0}^{-1}$ considered as an operator on $L^2_0(\mu)$:
\[
\left\| \cL_{\rm pert} \cL_{\xi,0}^{-1} \right\|_{\sB(L^2_0(\mu))} \leq \sqrt{\frac{\beta}{\xi} \left\| \cL_{\xi,0}^{-1} \right\|_{\sB(L^2_0(\mu))}} \leq \frac{\sqrt{\beta K}}{\min(1,\xi)}. 
\]  
Let us denote by $r$ the spectral radius of the operator appearing in the left hand-side of the previous inequality, namely
\[
r = \lim_{n \to +\infty}  \left\| \left[ \left(\cL_{\rm pert} \cL_{\xi,0}^{-1} \right)^* \right]^n \right\|_{\sB(L^2_0(\mu))}^{1/n},
\] 
The following result (given by~\cite[Theorem~5.2]{LS16}) provides an expression of $h_\tau$.

\begin{proposition}
  \label{thm:expression_f_xi}
 For $|\tau| < r^{-1}$, the unique invariant measure can be written as $\psi_\tau = h_\tau\mu$, where 
 $h_\tau \in L^2(\mu)$ admits the following expansion in powers of~$\tau$:
  \begin{equation}
    \label{eq:expansion_psi_xi_general}
    h_\tau =  \left( 1+\tau \left(\cL_{\rm pert} \cL_{\xi,0}^{-1}\right)^*  \right)^{-1} \mathbf{1} 
    = \left( 1 + \sum_{n=1}^{+\infty} (-\tau)^n \left[ \left(\cL_{\rm pert} \cL_{\xi,0}^{-1} \right)^* \right]^n \right) \mathbf{1}.
  \end{equation}
\end{proposition}

Since $0 \leq r \leq \left\| \cL_{\rm pert} \cL_{\xi,0}^{-1} \right\|_{\sB(L^2_0(\mu))}$, it holds
\begin{equation}
  \label{eq:scaling_r_xi}
  \frac{1}{r} \geq \frac{\min(1,\xi)}{\sqrt{\beta K}}.
\end{equation}
This shows that the upper limit on~$\tau$ for the validity of the expansion~\eqref{eq:expansion_psi_xi_general} should be of order $\min(1,\xi)$. This is consistent with physical intuition: when the fluctuation/dissipation is small, the external forcings which can be sustained by the system are at most of the same order of magnitude than the dissipation mechanism; while for large fluctuation/dissipation mechanisms, external forcings of the same order of magnitude as $\nabla U$ can be sustained. Let us mention that it is not clear whether the condition~\eqref{eq:scaling_r_xi} really is a necessary one. For extremely large forcings~$\tau$, we expect that the invariant measure will be quite different from~$\mu$, so that it is not a surprise that the perturbative approach of Proposition~\ref{thm:expression_f_xi} fails as such.

As an application of the power expansion provided in Proposition~\ref{thm:expression_f_xi}, let us recall one way to obtain \eqref{eq:ER}-\eqref{eq:diffusivity}. Note first that~\eqref{eq:diffusivity} can be rewritten as
\[
D = \int_0^{+\infty} \mathbb{E}_0\left[\frac{p_t}{m} \otimes \frac{p_0}{m}\right] \, dt.
\]
Next, in view of~\eqref{eq:expansion_psi_xi_general} and using the following equality in the sense of bounded operators on $L^2_0(\mu)$ (relying for instance on the convergence results of Section~\ref{sec:exp_cv} with $\tau = 0$, which ensure that the time integral is well defined):
\[
-\cL_{\xi,0}^{-1} = \int_0^{+\infty} \rme^{t \cL_{\xi,0}} \, dt,
\]
it follows that
\[
\frac d{d\tau} v_\tau \Big|_{\tau = 0} = -\int_\cE \frac{p}{m} \left(\cL_{\xi,0}^*\right)^{-1} \cL_{\rm pert}^* \mathbf{1} \, d\mu = -\beta\int_\cE \cL_{\xi,0}^{-1}\left(\frac{p}{m}\right) \frac{p}{m} \cdot F \, d\mu = \beta\int_0^{+\infty} \mathbb{E}_0\left[\frac{p_t}{m} \left(\frac{p_0}{m}\cdot F\right)\right] dt.
\]
This is indeed~\eqref{eq:ER}.

\section{Exponential convergence of the law}
\label{sec:exp_cv}

From the convergence result provided in Proposition~\ref{prop:unif_cv_cL_eta} and the stationary Fokker-Planck equation~\eqref{eq:FP_L2}, it is expected that $\rme^{t \cLxt^*}f$ converges to $h_\tau$ for any initial density $f \in L^2(\mu)$ such that 
\[
\int_\cE f \, d\mu = 1, \qquad f \geq 0.
\]
We state several such results in this section: the first two are based on the hypocoercive approach presented in~\cite{Villani09} and used also in~\cite{HP08} for equilibrium Langevin dynamics (see Section~\ref{sec:std_hypoc} for a result in a degenerate $H^1(\mu)$ norm and the related convergence in $L^2(\mu)$ after hypoelliptic regularization), while the third one follows from the more direct approach of~\cite{DMS09,DMS15} (see Section~\ref{sec:direct_L2}). The main point is that the latter approach is the only one which allows to state exponential convergence results consistent with the upper bound~\eqref{eq:scaling_r_xi} on the series expansion of the invariant measure, while the more traditional hypocoercive approach is limited to small forcings $\tau = \mathrm{O}(\xi^{-1})$ for large frictions. 

Although the convergence results are stated here for probability densities, they can in fact be obtained for general elements of $L^2(\mu)$ (\textit{i.e.} the functions under consideration need not be non-negative and be of integral~1 with respect to~$\mu$). We write them however for elements of $L^2(\mu)$ of mass~1, introducing to this end
\[
L^2_1(\mu) = \left\{ f \in L^2(\mu), \ \int_\cE f \, d\mu = 1 \right\}.
\]
Convergence results for arbitrary elements of $L^2(\mu)$ follow by an appropriate renormalization.

\subsection{Standard hypocoercive approach}
\label{sec:std_hypoc}

We start by stating a convergence result in a norm finer than the $L^2(\mu)$ norm, but coarser than the $H^1(\mu)$ norm usually considered in the theory of hypocoercivity. Instances of such degenerate norms can be found in~\cite{Talay02,Baudoin13,OL15,Baudoin16}. More precisely, we consider the Hilbert space
\[
\mathcal{H} = \left\{ f \in L^2(\mu), \ (\nabla_p+\nabla_q)f \in L^2(\mu) \right\},
\]
endowed with the scalar product (for $a>0$)
\[
\la f,g \right\rangle_a = \la f,g \ra + a \la (\nabla_p+\nabla_q)f, (\nabla_p+\nabla_q)g \ra.
\]
The associated scalar product is denoted by $\|\cdot\|_a$. The precise convergence result is the following.

\begin{theorem}
  \label{thm:hypo_H1_like}
  There exist $\delta_* > 0$ and a continuous function $a:(0,+\infty) \to (0,+\infty)$ satisfying
  \[
  \lim_{\xi \to 0} \frac{a(\xi)}{\xi} = \overline{a}_0 > 0, 
  \qquad
  \lim_{\xi \to +\infty} \xi a(\xi) = \overline{a}_\infty > 0, 
  \]
  such that, for any $\delta \in [0,\delta^*]$, there is $\overline{\lambda}_\delta>0$ for which, for all $\xi \in (0,+\infty)$ and $\tau \in [-\delta \mx, \delta \mx]$,
  \begin{equation}
    \label{eq:cv_H1_like}
    \forall f \in \mathcal{H} \cap L^2_1(\mu), \ \forall t \geq 0, 
    \qquad 
    \left\| \rme^{t \cLxt^*}f - h_\tau \right\|_{a(\xi)} \leq \rme^{-\overline{\lambda}_\delta \mx t} \| f - h_\tau \|_{a(\xi)}.
  \end{equation}
  Moreover, $\overline{\lambda}_\delta = \overline{\lambda}_0 + \mathrm{O}(\delta)$. 
\end{theorem}

The proof of this result can be read in Section~\ref{sec:proof_H1_like}. A careful inspection of the final argument based on an asymptotic analysis of the key matrix inequality~\eqref{eq:condition_matrix_H1} would allow to give more precise expressions of $\overline{a}_0,\overline{a}_\infty$ and $\overline{\lambda}_\delta$ as a function of the parameters of the dynamics. Let us also emphasize that, in the degenerate norm we consider, there is no prefactor on the right-hand side of~\eqref{eq:cv_H1_like}, contrarily to convergence theorems stated in $H^1(\mu)$ for which the prefactor degenerates in the limits $\xi \to 0$ or $\xi \to +\infty$. Let us finally mention that the way we formulate our result, by considering $\tau \in [-\delta \mx, \delta \mx]$ (rather than $\tau \in [-\delta_* \mx, \delta_* \mx]$), allows to emphasize that the convergence rate has some uniformity with respect to~$\tau$.

\medskip

By hypoelliptic regularization, the convergence result of Theorem~\ref{thm:hypo_H1_like} can be transferred to $L^2(\mu)$ (see Section~\ref{sec:proof_hypo_L2_like} for the proof). 

\begin{theorem}
  \label{thm:hypo_L2_like}
  There exist $C,\delta_* > 0$ such that, for any $\delta \in [0,\delta^*]$, there is $\overline{\lambda}_\delta>0$ for which, for all $\xi \in (0,+\infty)$ and all $\tau \in [-\delta \mx, \delta \mx]$, 
  \begin{equation}
    \label{eq:cv_L2_like}
    \forall f \in L^2_1(\mu), \ \forall t \geq 0, 
    \qquad 
    \left\| \rme^{t \cLxt^*}f - h_\tau \right\|_\Ltwo \leq C \rme^{-\overline{\lambda}_\delta \mx t} \| f - h_\tau \|_\Ltwo.
  \end{equation}
\end{theorem}

The convergence rate $\overline{\lambda}_\delta$ is the same as in Theorem~\ref{thm:hypo_H1_like}. Note also that the prefactor $C$ is independent of~$\delta$ since the regularization properties of the dynamics can be shown to be uniform with respect to~$\tau$ (see Proposition~\ref{prop:hypoelliptic_reg} in Section~\ref{sec:proof_hypo_L2_like}).

We conclude by emphasizing an important restriction of the results in this section: in the large friction limit, only very small forcings~$\tau$, of order~$1/\xi$ are allowed, whereas, in view of~\eqref{eq:scaling_r_xi}, the density $h_\tau$ is well defined for values of~$\tau$ of order~1. This restriction can however be bypassed by the more direct approach from~\cite{DMS09,DMS15}, as made precise in Section~\ref{sec:direct_L2}.

\subsection{Direct $L^2$ estimates}
\label{sec:direct_L2}

We state a convergence result similar to Theorem~\ref{thm:hypo_L2_like}, with the important difference that the too stringent restriction on the upper bound of $\tau$ is removed; see Section~\ref{sec:proof_L2_DMS} for the proof.

\begin{theorem}
  \label{thm:L2_DMS}
  There exist $C,\delta_* > 0$ such that, for any $\delta \in [0,\delta^*]$, there is $\overline{\lambda}_\delta>0$ for which, for all $\xi \in (0,+\infty)$ and all $\tau \in [-\delta \min(\xi,1), \delta \min(\xi,1)]$,
  \begin{equation}
    \label{eq:cv_L2_like_DMS}
    \forall f \in L_1^2(\mu), \ \forall t \geq 0, 
    \qquad 
    \left\| \rme^{t \cLxt^*}f - h_\tau \right\|_\Ltwo \leq C \rme^{-\overline{\lambda}_\delta \mx t} \| f - h_\tau \|_\Ltwo.
  \end{equation}
  Moreover, $\overline{\lambda}_\delta = \overline{\lambda}_0 + \mathrm{O}(\delta)$.
\end{theorem}

As in Section~\ref{sec:std_hypoc}, the convergence rate~$\overline{\lambda}_\delta$ can be quantified in terms of the various parameters of the dynamics by optimizing the smallest eigenvalue of a matrix (see~\eqref{eq:key_estimate_DMS} in the proof). It would also be possible to obtain a contraction result on $\| \rme^{t \cLxt^*}f - \rme^{t \cLxt^*}g\|_\Ltwo$ for two elements $f,g \in L^2_1(\mu)$, from which the existence and uniqueness of an invariant measure characterized by $h_\tau$ can be deduced as in~\cite{BHM16}. 

\section{Numerical estimation of the spectral gap}
\label{sec:num}

Let us first relate the exponential decay of the semigroup $\rme^{t \cLxt^*}$ with the spectral gap of $\cL^*_{\xi,\tau}$ in $L^2(\mu)$; and in fact the spectral gap of $\cL_{\xi,\tau}$. We fix $\xi \in (0,+\infty)$ and $\tau \in [-\delta \min(\xi,1), \delta_* \min(\xi,1) ]$, where $0 \leq \delta \leq \delta_*$ is defined in Theorem~\ref{thm:L2_DMS}. We first note that the decay estimate~\eqref{eq:cv_L2_like_DMS} implies that for all $f,g \in L^2(\mu)$,
\[
\begin{aligned}
  \left| \la \rme^{t \cLxt}g - \la g, h_\tau \ra, f \ra \right| & = \left| \la \rme^{t \cLxt}g, f \ra - \la g, h_\tau\ra  \int_\cE f \, d\mu \right| = \left| \la g, \rme^{t \cLxt^*}f - h_\tau  \int_\cE f \, d\mu\ra \right| \\
  & \leq  \| g \|_{L^2(\mu)} \left\|  \rme^{t \cLxt^*}f - h_\tau \int_\cE f \, d\mu\right\|_{L^2(\mu)} \\
  & \leq C \rme^{-\overline{\lambda}_\delta \min(\xi,\xi^{-1})t} \| g \|_{L^2(\mu)} \left\| f -  h_\tau\int_\cE f \, d\mu \right\|_{L^2(\mu)}\\
  & \leq  C \rme^{-\overline{\lambda}_\delta \min(\xi,\xi^{-1})t} \| g \|_{L^2(\mu)}\left( 1 + \| h_\tau\|_{L^2(\mu)}\right) \| f \|_{L^2(\mu)},
\end{aligned}
\]
using a Cauchy--Schwarz inequality for the bound $\|f\|_{L^1(\mu)} \leq \| f \|_{L^2(\mu)}$. By taking the supremum over~$f \in L^2(\mu)$ with $\|f\|_\Ltwo = 1$, it follows that
\begin{equation}
  \label{eq:decay_semigroup_observables}
  \left\| \rme^{t \cLxt} g - \la g, h_\tau \ra \right\|_\Ltwo \leq C \rme^{-\overline{\lambda}_\delta \min(\xi,\xi^{-1})t} \| g \|_{L^2(\mu)}\left( 1 + \| h_\tau\|_{L^2(\mu)}\right).
\end{equation}

Let us introduce the subspace of observables with average~0 with respect to the invariant measure of the nonequilibrium system:
\[
L_{0,\tau}^2(\mu) = \left\{ g \in L^2(\mu) \,\left| \, \la g, h_\tau \ra = 0 \right.\right\}.
\]
A simple computation shows that this space is stable by $\rme^{t \cLxt}$. The decay estimate~\eqref{eq:decay_semigroup_observables} then shows that the following equality holds in $L^2_{0,\tau}(\mu)$: for all $z \in \mathbb{C}$ such that $\mathrm{Re}(z) < \overline{\lambda}_\delta \min(\xi,\xi^{-1})$, 
\[
\left(z+\cLxt\right)^{-1} = -\int_0^{+\infty} \rme^{t \cLxt} \, \rme^{z t} \, dt,
\]
with
\[
\left\|\left(z+\cLxt\right)^{-1} \right\|_{\sB(L^2_{0,\tau}(\mu))} \leq  \frac{C\left( 1 + \| h_\tau\|_{L^2(\mu)}\right)}{\overline{\lambda_\delta} \min(\xi,\xi^{-1})-\mathrm{Re}(z)}.
\]
This means that the spectral gap $\gamma(\xi,\tau)$ of the generator $\cLxt$ on $L^2_{0,\tau}(\mu)$, defined as
\[
\gamma(\xi,\tau) = \min \left\{ \mathrm{Re}(z), \, z \in \sigma(-\cLxt) \backslash \{0\} \right\},
\]
is bounded from below as 
\begin{equation}
  \label{eq:lower_bound_sg}
  \gamma(\xi,\tau) \geq \overline{\lambda}_\delta \min(\xi,\xi^{-1}).
\end{equation}
We show in Section~\ref{sec:Galerkin} how to approximate the spectral gap using a Galerkin discretization of the generator $\cLxt$. We then study in Section~\ref{sec:num_res} the relevance of the lower bound~\eqref{eq:lower_bound_sg} from a numerical viewpoint, and check that it is in fact sharp.

\subsection{Discretization by a Galerkin procedure}
\label{sec:Galerkin}

We consider $d=1$ for simplicity, although the discretization procedure described below can be extended to any arbitrary dimension by a tensorization argument, with of course the caveat that the computational cost of the method explodes. The Galerkin discretization we use relies on a tensor product of Fourier modes in space and Hermite functions for the momenta, see for instance~\cite{LPK13,RST16,RS17} for previous similar discretizations as well as~\cite{Risken} for a seminal presentation of such approaches. More precisely, we introduce the following discretization basis for $n \in \{0,\dots,N\}$ and $k\in \{-K,\dots,K\}$:
\[
\psi_{nk}(q,p)= G_k(q) H_n(p), \qquad G_k(q) = \sqrt{\frac{Z_\nu}{2\pi}} \, \rme^{\ri kq+ \beta U(q)/2},
\]
where $H_n$ is the Hermite polynomial of degree $n$:
\[
H_n(p)\,=\,\frac{(-1)^n}{n!}\,\left(\sqrt{\frac m\beta}\right)^n \rme^{\beta p^2/(2m)} \frac{\dd^n}{\dd p^n}\left( \rme^{-\beta p^2/(2m)} \right),
\]
and 
\[
Z_\nu = \int_0^{2\pi} \rme^{-\beta U}
\]
is the normalization constant of the marginal of~$\mu$ in the position variable. Recall that the Hermite polynomials are orthonormal on $L^2(\kappa)$, where $\kappa(dp)$ is the marginal of~$\mu$ in the momentum variable:
\begin{equation}
  \label{eq:kappa}
  \kappa(dp) = \int_\cD \mu(dq \,dp) = \sqrt{\frac{\beta}{2\pi m}} \rme^{-\beta p^2/(2m)} \, dp.
\end{equation}
Note that the family $\{\psi_{nk}(q,p)\}_{0 \leq n \leq N, -K \leq k\leq K}$ is orthonormal in $L^2(\mu)$, and spans $L^2(\mu)$ in the limit $N \to +\infty$ and $K \to +\infty$. 

The generator $\cLxt$ is represented in this basis by a matrix with elements
\begin{equation}
  \label{eq:matrix_approx_cLxt}
  \left[L_{\xi,\tau}\right]_{n'k',nk} = \la \psi_{n'k'},\cLxt \psi_{nk} \ra.
\end{equation}
These matrix elements are easily computed using the following properties of Hermite functions:
\[
\begin{aligned}
\partial_p H_n(p) & = \sqrt{\frac{\beta n}{m}} H_{n-1}(p)\,,\\
\partial_p^* H_n(p) & = \frac{\beta p}{m} H_{n}(p)\,-\partial_p H_{n-1}(p) = \sqrt{\frac{\beta(n+1)}{m}}\,H_{n+1}(p).
\end{aligned}
\]
In particular, Hermite polynomials are a complete set of eigenfunctions for $\cLFD$ on $L^2(\kappa)$, and $\cLFD H_n = -n H_n/m$. In addition, the action of derivatives on the Fourier modes can be evaluated as
\[
\la G_{k'}, \partial_q G_k\right\rangle_{L^2(\nu)} = \ri k \delta_{k,k'} + \frac{\beta}{2} u_{k-k'}, 
\qquad 
u_K = \frac{1}{2\pi} \int_0^{2\pi} U'(q) \, \rme^{\ri Kq} \, dq.
\]
Finally, note also that the Hamiltonian part of the generator can be rewritten as
\[
\cLham = \frac1\beta \left(\partial_q \partial_p^* - \partial_p \partial_q^* \right).
\]
Therefore, 
\[
\begin{aligned}
\la \psi_{n'k'}, \cLham \psi_{nk} \ra & = \frac1\beta\left( \la G_{k'}, \partial_q G_{k}\right\rangle_{L^2(\nu)} \la H_{n'}, \partial_p^* H_n\right\rangle_{L^2(\kappa)} - \overline{\la G_{k}, \partial_q G_{k'}\right\rangle_{L^2(\nu)}} \la H_{n'}, \partial_p H_n\right\rangle_{L^2(\kappa)} \right) \\
& = \sqrt{\frac{\beta(n+1)}{m}} \delta_{n',n+1}\left[\frac{\ri k}{\beta} \delta_{k,k'} + \frac{1}{2} u_{k-k'}\right] + \sqrt{\frac{\beta n}{m}} \delta_{n',n-1} \left[ \frac{\ri k}{\beta} \delta_{k,k'} - \frac{1}{2} u_{k-k'}\right].
\end{aligned}
\]
Moreover, with $F=1$, 
\[
\la \psi_{n'k'},\cLpert \psi_{nk} \ra = \sqrt{\frac{\beta n}{m}} \delta_{n',n-1} \delta_{k,k'}.
\]
In conclusion,
\[
\begin{aligned}
\la \psi_{n'k'}, \cLxt \psi_{nk} \ra & = \sqrt{\frac{\beta(n+1)}{m}} \delta_{n',n+1}\left[\frac{\ri k}{\beta} \delta_{k,k'} + \frac{1}{2} u_{k-k'}\right] + \sqrt{\frac{\beta n}{m}} \delta_{n',n-1} \left[ \left(\frac{\ri k}{\beta} + \tau\right)\delta_{k,k'} - \frac{1}{2} u_{k-k'}\right] \\
& - \frac{\xi n}{m} \delta_{n,n'} \delta_{k,k'}.
\end{aligned}
\]
The spectral gap of $\cLxt$ is approximated by the spectral gap of the matrix with elements~\eqref{eq:matrix_approx_cLxt}. There are currently no error estimates on the spectral approximation provided by this procedure, while in contrast there is a large body of literature on the Galerkin approximation of operators defined by quadratic forms (see for instance~\cite{Chatelin} and references therein). Some preliminary steps in this direction are however provided in~\cite{RS17} in the form of error estimates on the discretization of solutions of Poisson equations $-\cLxt \Phi = f$.


Although the extension of the Galerkin procedure poses no difficulties from a conceptual viewpoint, it is strongly limited by the size of the matrices to be considered, which, for a standard tensorized basis, are of dimension $(N+1)^d (2K+1)^d$. The two-dimensional case is therefore already challenging. In order to make the method more efficient, it would be necessary to use the sparsity of the matrices under consideration, devise better bases than tensorized ones by using tensor formats~\cite{Hackbusch}, relying on preconditionning strategies, etc. This work is in progress.

\subsection{Numerical results}
\label{sec:num_res}

We choose $U(q) = U_0(1 - \cos{q})$, which corresponds to the so-called rotor model. The partition function $Z_\nu$ can be explicitly computed as
\[
Z_\nu = 2\pi \rme^{-\beta U_0} I_0(\beta U_0),
\]
where $I_0(a) = (2\pi)^{-1}\int_0^{2\pi} \rme^{a \cos(q)} \, \dd q$ is the modified Bessel function of the first kind. In addition,
\[
u_{k-k'} = -\frac{\ri U_0}{2} \left(\delta_{k',k+1} - \delta_{k',k-1} \right).
\]
We first show in Section~\ref{sec:cv_num} that the spectral gap can be estimated with the procedure described in Section~\ref{sec:Galerkin}, by comparing the computed spectral gap to known analytical values in the case when $U_0 = 0$ (see~\cite{Kozlov89}) and studying the convergence of the spectral gap in the limit $N,K \to +\infty$. We next present in Section~\ref{sec:sg} plots of the spectral gap a function of the friction, for various values of the external forcing~$\tau$ and magnitude~$U_0$ of the potential energy.

\subsubsection{Checking the convergence}
\label{sec:cv_num}

We first study the convergence of the spectral gap with respect to the basis sizes $N,K$. We choose $N = 2K$ in all our simulations, as well as $m=\beta=1$. As can be seen from the numerical results presented in Figure~\ref{fig:cv} and as confirmed by extensive simulations not reported here, the convergence is faster for larger values of the friction~$\xi$. In the results reported below, we consider the spectral gap to be converged when the relative variation between the current estimation of the spectral gap for a given basis size $K=N/2$ and the average of the last three ones corresponding to values of $K-1,K-2,K-3$ is lower than $10^{-3}$. Let us also emphasize that, when the values of $\tau$ are too large in the case of small frictions, we do not observe convergence for reasonable values of $K$ (about~20). In fact, for too large forcings, we observe that the spectral gap is negative for the smallest values of~$K$, but then becomes positive again for sufficiently large~$K$; although this stabilization to a positive value may occur for prohibitely large values of $K$ for small~$\xi$. 

\begin{figure}
\begin{center}
\includegraphics[width=0.48\textwidth]{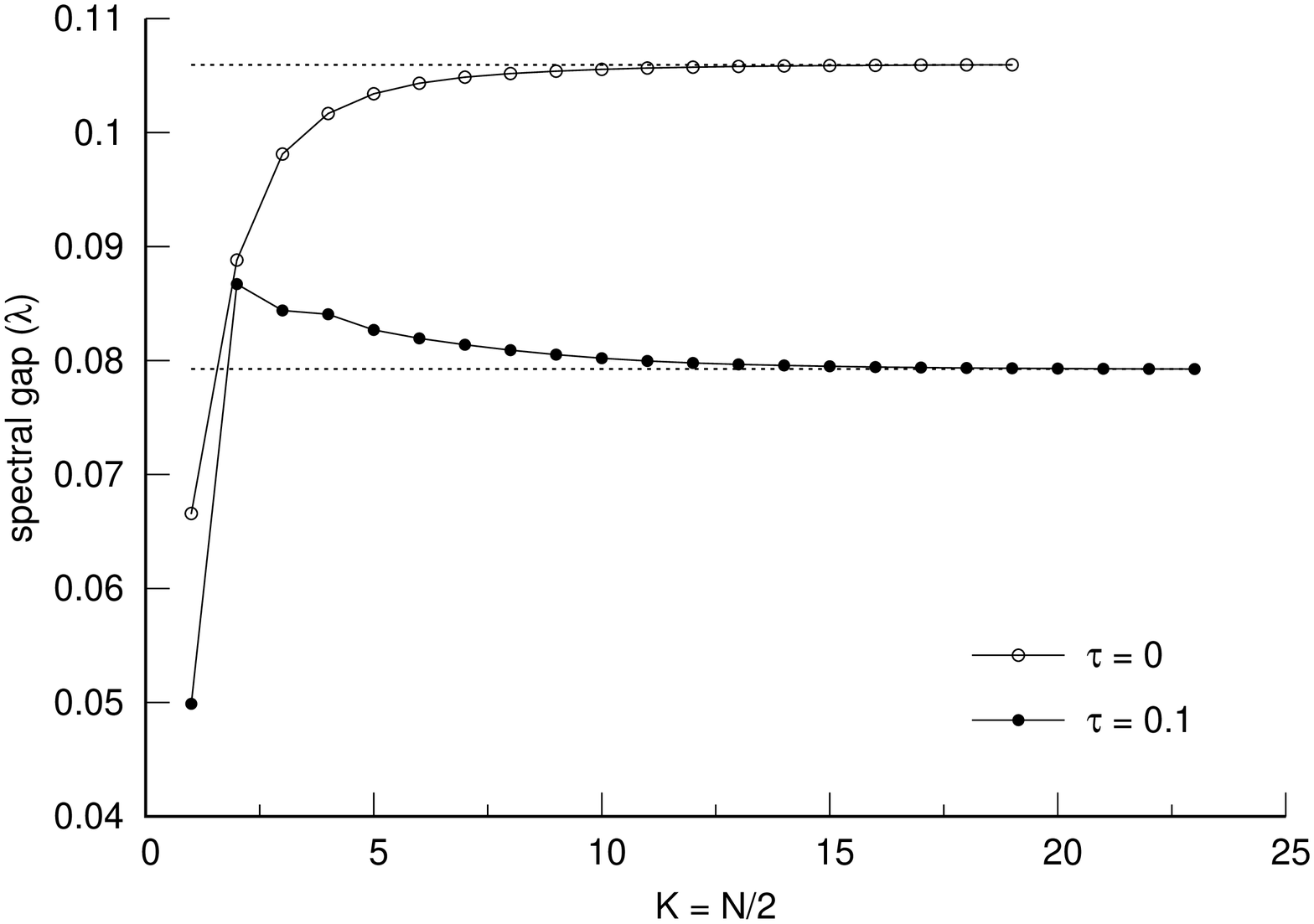}
\includegraphics[width=0.48\textwidth]{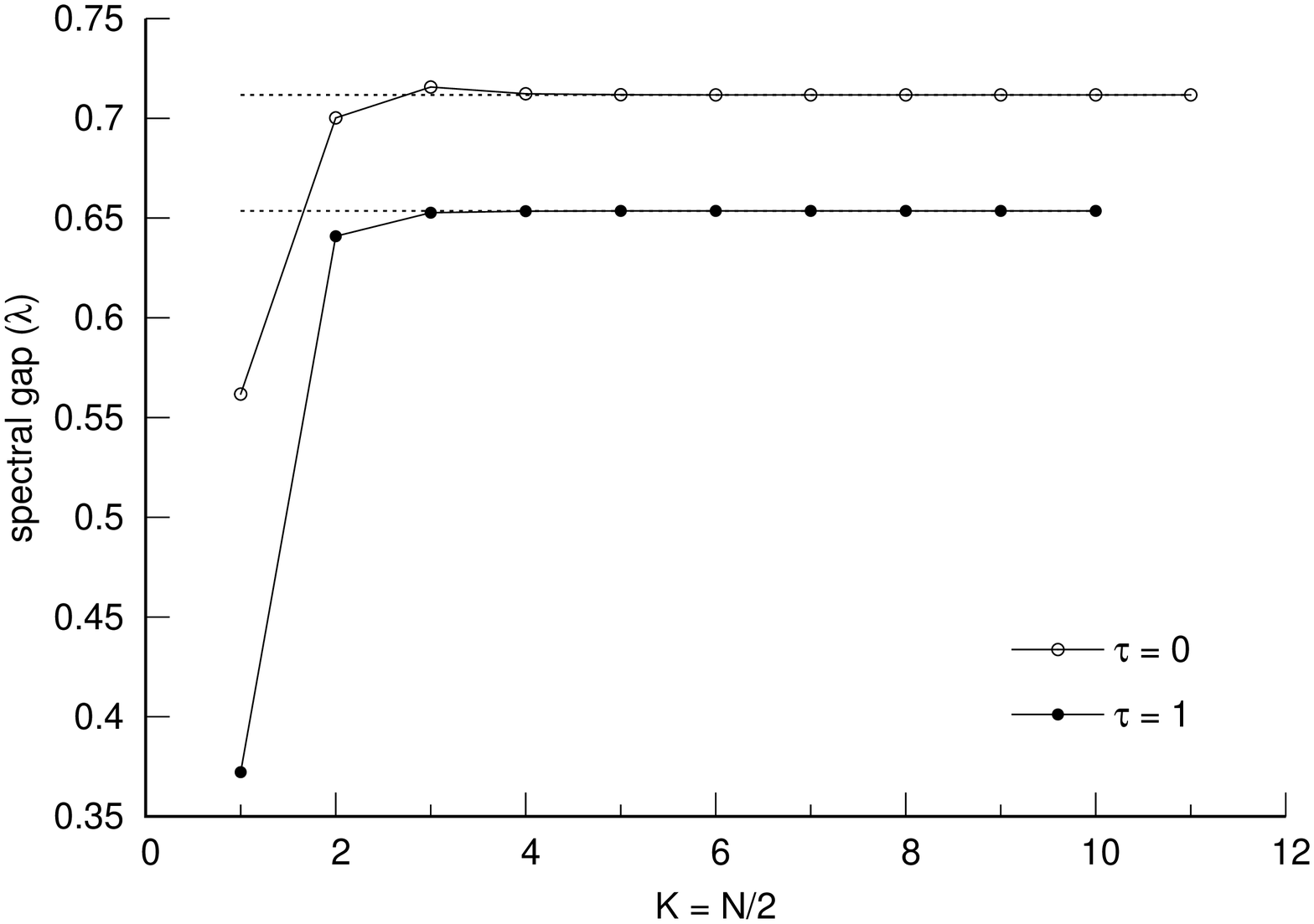}
\end{center}
\caption{\label{fig:cv} Convergence of the estimated spectral gap as a function of the basis sizes $K = N/2$, for $U_0 = 1$. Left: $\xi = 0.1$. Right: $\xi = 1$.}
\end{figure}

As a consistency check, we next verify that the spectral gap predicted by the Galerkin method agrees with the one which can be analytically computed in the case when $U_0=0$. More precisely, the results of~\cite{Kozlov89} show that the eigenvalues are
\[
\lambda_{n,k} = - \frac{n \xi}{m} - \frac{k^2}{\beta \xi},
\]
so that the spectral gap is
\begin{equation}
\label{eq:sg_Kozlov}
\gamma(\xi,0) = \min\left(\frac{\xi}{m},\frac{1}{\beta \xi}\right).
\end{equation}
The crossover from one eigenvalue branch to the other occurs at $\xi = \sqrt{m/\beta}$. Figure~\ref{fig:Kozlov} presents the eigenvalues which are numerically computed. They are in perfect agreement with the theoretical prediction~\eqref{eq:sg_Kozlov}.

\begin{figure}
\begin{center}
\includegraphics[width=0.7\textwidth]{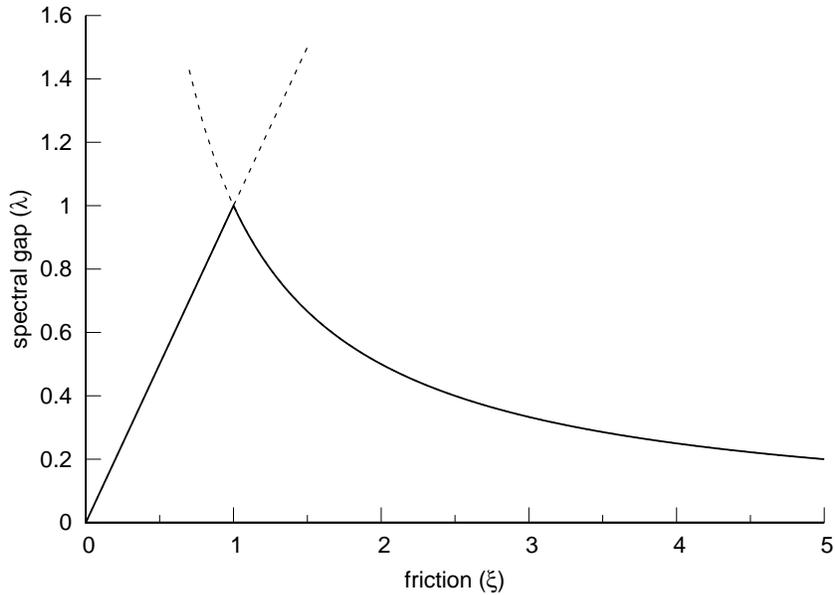}
\end{center}
\caption{\label{fig:Kozlov} Predicted spectral gap as a function of the friction~$\xi$ when $U_0=0$, $\beta=1$ and $m=1$ (solid line), and theoretical prediction~\eqref{eq:sg_Kozlov} (dashed lines; one of the lines represents $\xi$ the other one~$1/\xi$).}
\end{figure}

\subsubsection{Spectral gap as a function of the friction}
\label{sec:sg}

We report in Figures~\ref{fig:sg_tau_U1} and~\ref{fig:sg_tau_U1_zoom} the spectral gap as a function of the friction~$\xi$ for various values of the forcing~$\tau$ and the choice $U_0 = 1$ (with $m=1$ and $\beta=1$). The first point to mention is that the introduction of a potential smoothes out the sharp transition observed at $\xi = \sqrt{m/\beta} = 1$ when $U_0 = 0$ and $\tau = 0$ (recall Figure~\ref{fig:Kozlov}). As $\tau$ is increased, we were able to compute the spectral gap only for frictions above a certain treshold, roughly $\xi \geq |\tau|$. The spectral gap is increased with respect to the case~$\tau=0$ in a certain range of values of $\xi$ (roughly, $1 \leq \xi \leq 4$), and decreased for other ones (in particular $\xi \leq 1$). Forcings however have to be sufficiently strong in order for the perturbation to have some non-negligible impact (when $\tau = 0.1$, the perturbation on the spectral gap is visible only around $\xi = 2$). Lastly, we observe that the spectral gap decreases as $\Gamma_\tau/\xi$ for large~$\xi$ (for some prefactor $\Gamma_\tau > 0$) whatever the value of $\tau$ we considered, with $\Gamma_\tau$ decreasing as $\tau$ increases.

\begin{figure}
\begin{center}
\includegraphics[width=0.6\textwidth]{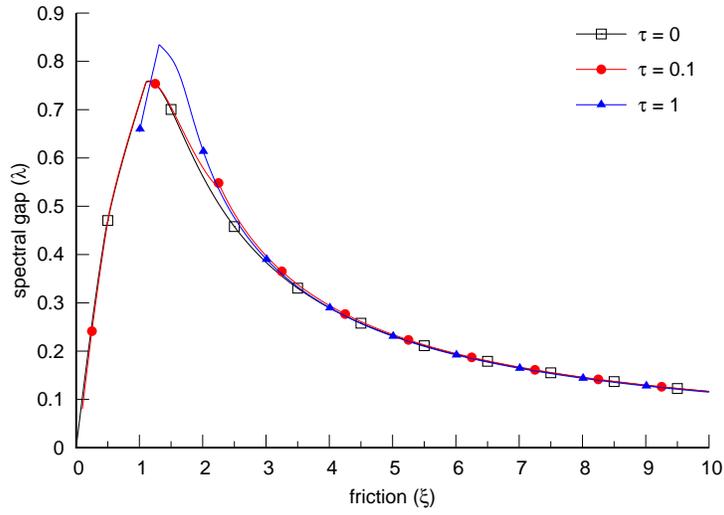}
\end{center}
\caption{\label{fig:sg_tau_U1} Spectral gap as a function of~$\xi$ for $\tau = 0, 0.1, 1$ when $U(q) = 1-\cos(q)$.}
\end{figure}

\begin{figure}
\begin{center}
\includegraphics[width=0.48\textwidth]{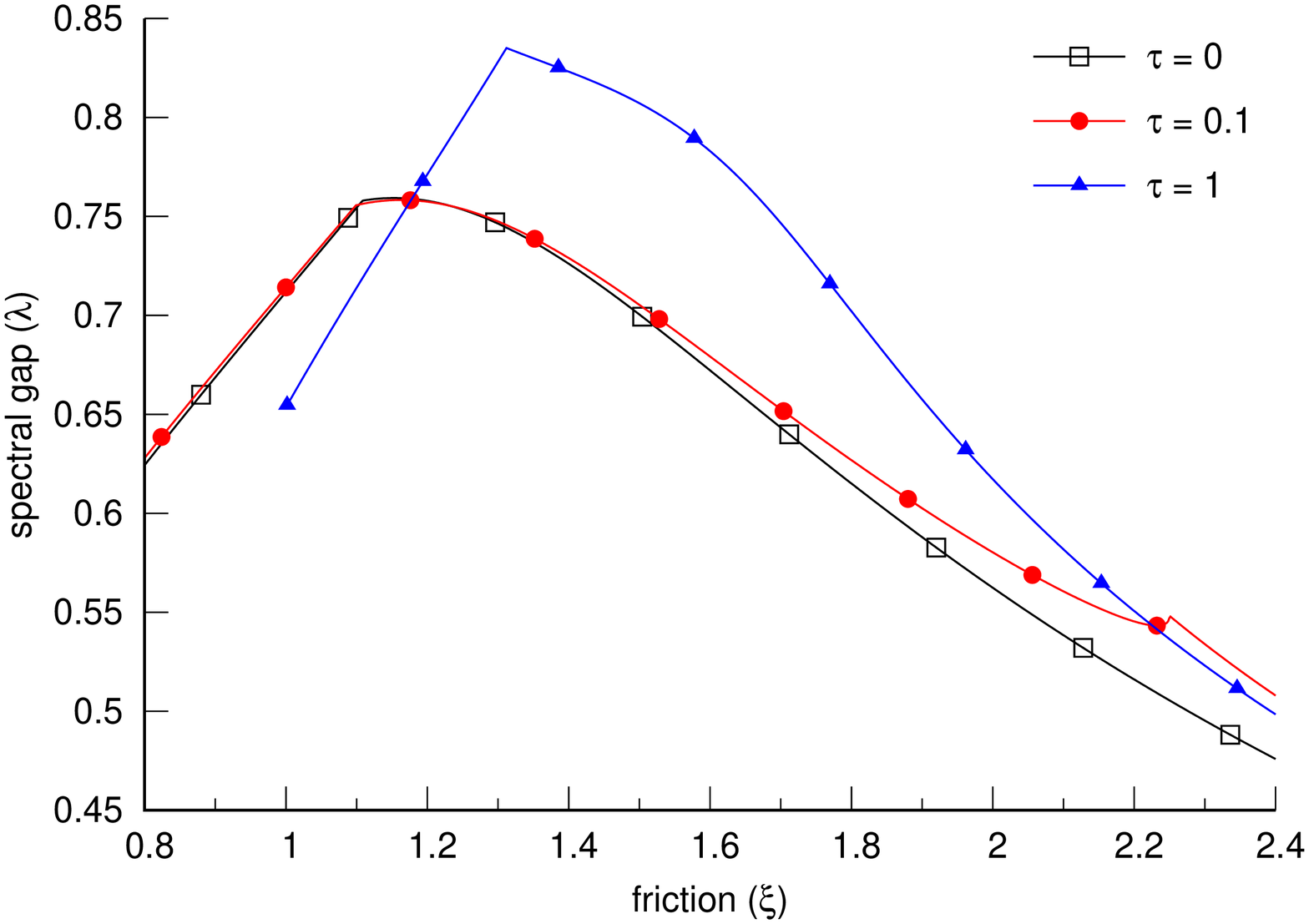}
\includegraphics[width=0.48\textwidth]{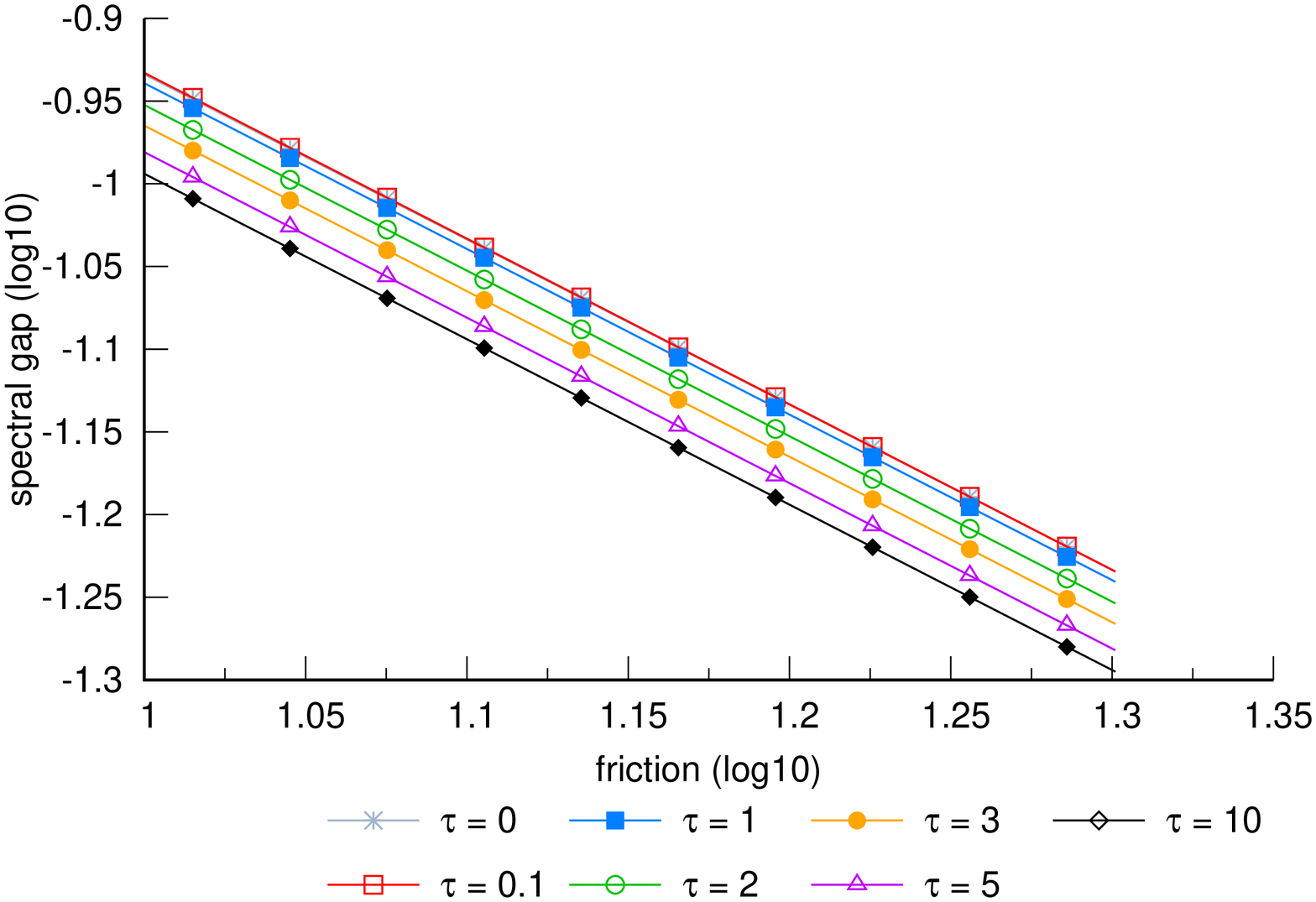}
\end{center}
\caption{\label{fig:sg_tau_U1_zoom} Zoom of Figure~\ref{fig:sg_tau_U1} around $\xi = 1$ (Left) and large $\xi$ (Right; logarithmic scale).}
\end{figure}

We next study the influence of the potential by computing the spectral gap for various values of $\tau,\xi$ when the potential is $U(q) = U_0(1-\cos(q))$. The results for $U_0 = 0.1$ and $U_0 = 10$ are reported in Figure~\ref{fig:sg_U}. Note that the effect of the perturbation is less visible on these pictures: the spectral gaps are much closer to the ones corresponding to $\tau = 0$. For $U_0$ small, this is related to the fact that the spectral gap is independent of~$\tau$ when $U_0 = 0$ (see~\cite{Iacobucci} for further precisions), a feature which approximately persists for small but non-zero values of $U_0$. For large~$U_0$, the perturbation $\tau F$ is dominated by the gradient part of the force $U'(q)$; hence it has less impact on the spectral gap.

\begin{figure}
\begin{center}
\includegraphics[width=0.48\textwidth]{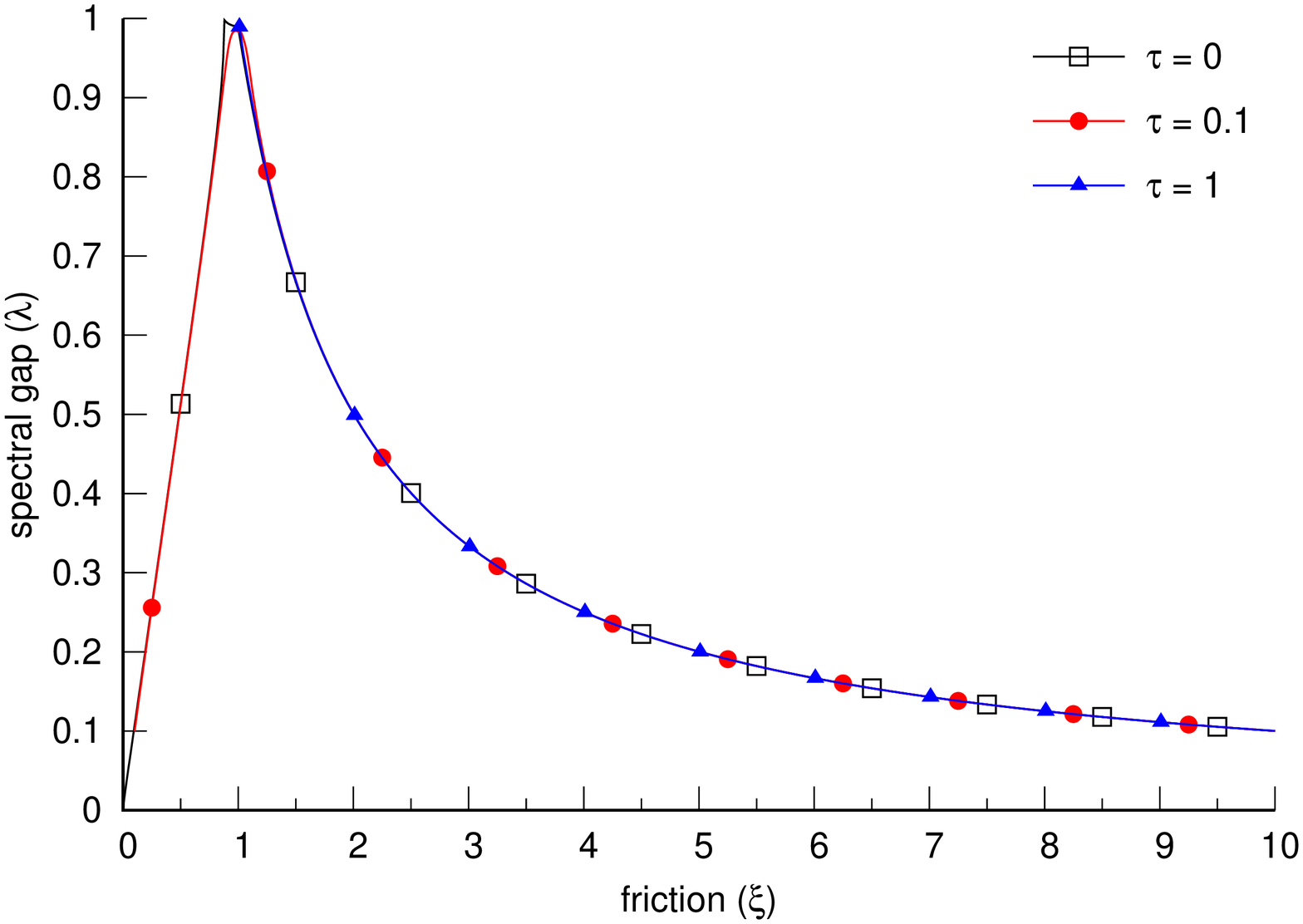}
\includegraphics[width=0.48\textwidth]{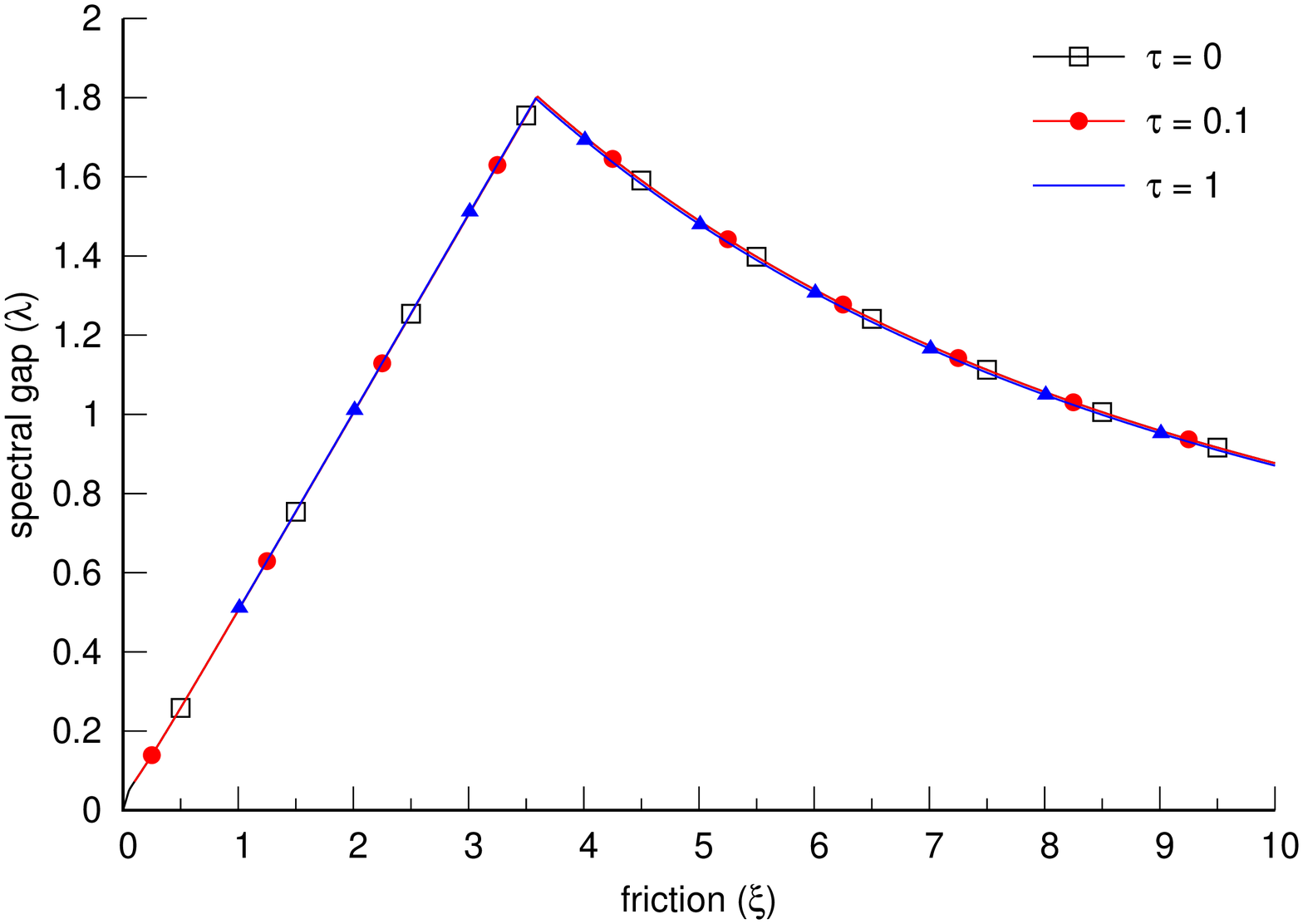}
\end{center}
\caption{\label{fig:sg_U} Spectral gap as a function of~$\xi$ for $\tau = 0, 0.1, 1$ when $U(q) = U_0(1-\cos(q))$. Left: $U_0 = 0.1$. Right: $U_0 = 10$.}
\end{figure}

\section{Proofs of the results}
\label{sec:proofs}

We introduce the marginal measure in the position variable: 
\[
\nu(dq) = \int_{\mathbb{R}^d} \mu(dq\,dp) = Z_\nu^{-1} \rme^{-\beta U(q)} \, dq.
\]
Recall also the definition~\eqref{eq:kappa} for the marginal~$\kappa(dq)$ in the momentum variable. These measures satisfy the following Poincar\'e inequalities:
\begin{equation}
  \label{eq:Poincare_nu}
  \forall \varphi \in H^1(\nu)\cap L^2_0(\nu), 
  \qquad 
  \| \varphi \|_{L^2(\nu)} \leq \frac{1}{K_\nu} \| \nabla_q \varphi  \|_{L^2(\nu)}, 
\end{equation}
and
\begin{equation}
  \label{eq:Poincare_kappa}
  \forall \phi \in H^1(\kappa)\cap L^2_0(\kappa), 
  \qquad
  \| \phi \|_{L^2(\kappa)} \leq \frac{1}{K_\kappa} \| \nabla_p \phi  \|_{L^2(\kappa)}. 
\end{equation}
In fact, $K_\kappa = \sqrt{\beta/m}$. The measure $\mu$ therefore also satisfies a Poincar\'e inequality:
\begin{equation}
  \label{eq:Poincare_total}
  \forall g \in H^1(\mu) \cap L^2_0(\mu), \qquad \| g \|_\Ltwo^2 \leq \frac{1}{K_\nu^2} \| \nabla_q g \|^2_\Ltwo + \frac{1}{K_\kappa^2} \| \nabla_p g \|^2_\Ltwo.
\end{equation}
This also implies 
\begin{equation}
  \label{eq:Poincare_total_2}
  \forall g \in H^1(\mu) \cap L^2_0(\mu), \qquad \| g \|_\Ltwo \leq \frac{1}{K_\nu} \| \nabla_q g \|_\Ltwo + \frac{1}{K_\kappa} \| \nabla_p g \|_\Ltwo.
\end{equation}

In all the proofs below, we start by fixing a function~$f$ in the Hilbert space under consideration, and define
\begin{equation}
  \label{eq:f(t)}
  f(t) = \rme^{t \cLxt^*} f - h_\tau.
\end{equation}
Note that, to simplify the notation, we omit the dependence of $f(t)$ on~$\xi,\tau$. We also consider that $\cLxt$ is defined on~$L^2(\mu)$, with domain $D(\cLxt) = \{ f \in L^2(\mu), \ \cLxt f \in L^2(\mu) \}$. In all cases under consideration, $f(t)$ has average~0 with respect to~$\mu$ for all times $t \geq 0$. 

We will also make use of the following remark: the smallest eigenvalue of the positive definite matrix 
\[
\mathcal{M} = \begin{pmatrix} a & b/2 \\ b/2 & c \end{pmatrix}
\]
is
\begin{equation}
  \label{eq:smallest_eig}
  \Lambda_-(\mathcal{M}) = \frac{a+c}{2} - \frac12 \sqrt{(a-c)^2 + b^2} = \frac{4ac - b^2}{a+c + \sqrt{(a-c)^2 + b^2}}.
\end{equation}

\subsection{Proof of Theorem~\ref{thm:hypo_H1_like}}
\label{sec:proof_H1_like}

Formally, 
\[
\frac{d}{dt}\left(\frac12 \| f(t)\|^2_{a(\xi)}\right) = \la f(t), \cLxt^* f(t) \raa.
\]
The exponential decay therefore follows by a Gronwall inequality provided
\[
\forall g \in D(\cLxt^*) \cap L^2_0(\mu), \qquad \la g, \cLxt^* g\raa \leq -\lambda(\xi,\tau) \| g \|^2_{a(\xi)},
\]
for some $\lambda(\xi,\tau) > 0$. Let us establish this inequality by considering a smooth and compactly supported function~$g$ with average~0 with respect to~$\mu$, and then conclude by density. By the computations recalled in~\cite[Section~2.2.3]{LS16} (which correspond to the equilibrium case $\tau = 0$),
\begin{equation}
  \label{eq:control_H1_std}
\la g, \cL_{\xi,0}^* g \raa \leq -X^T S(\xi) X - \frac{\xi a(\xi)}{\beta} \|(\nabla_p+\nabla_q)\nabla_p g\|^2_\Ltwo,
\end{equation}
with
\[
X = \begin{pmatrix} \| \nabla_p g\|_{L^2(\mu)} \\ \| \nabla_q g\|_{L^2(\mu)} \end{pmatrix}\!, 
\quad
S(\xi) = \begin{pmatrix} 
S_{pp}(\xi) \Id_d & S_{qp}(\xi) \Id_d/2\\
S_{qp}(\xi) \Id_d/2 & S_{qq}(\xi) \Id_d
\end{pmatrix},
\]
the elements of the matrix $S$ being
\[
S_{pp}(\xi) = \xi\left(\frac1\beta+\frac{a(\xi)}{m}\right) - a(\xi) \|\nabla^2 U\|_{L^\infty}, 
\quad 
S_{qp}(\xi) = -a(\xi)\left( \frac{1+\xi}{m} + \|\nabla^2U\|_{L^\infty}\right), 
\quad 
S_{qq}(\xi) = \frac{a(\xi)}{m}.
\]
In addition, the perturbation term can be bounded as
\[
\left|\la g, \cLpert^* g \raa\right| \leq \left|\la \cLpert g, g \ra\right| + a(\xi) \left|\la (\nabla_q+\nabla_p)g, (\nabla_q+\nabla_p) F \cdot \nabla_p^*g \raa\right|.
\]
Since $(\partial_{q_i}+\partial_{p_i})\partial_{p_j}^*g = \partial_{p_j}^*(\partial_{q_i}+\partial_{p_i})g + \delta_{ij} \beta g/m$, it follows that 
\[
\begin{aligned}
\left|\la g, \cLpert^* g \raa\right| & \leq \| g \|_\Ltwo\| \nabla_p g \|_\Ltwo + a(\xi) \| (\nabla_p+\nabla_q) g \|_\Ltwo\left(\| (\nabla_p+\nabla_q) \nabla_p g \|_\Ltwo + \frac{\beta}{m}\| g \|_\Ltwo\right) \\
& \leq \| g \|_\Ltwo \left[ \left(1+\frac{\beta a(\xi)}{m}\right) \| \nabla_p g \|_\Ltwo + \frac{\beta a(\xi)}{m}\| \nabla_q g \|_\Ltwo\right] + \eta a(\xi) \left( \| \nabla_p g \|_\Ltwo + \| \nabla_q g \|_\Ltwo \right)^2 \\
& \qquad + \frac{a(\xi)}{4\eta} \| (\nabla_p+\nabla_q) \nabla_p g \|_\Ltwo^2,
\end{aligned}
\]
for any $\eta>0$ by Young's inequality. Therefore, using~\eqref{eq:Poincare_total_2}, 
\begin{equation}
  \label{eq:control_H1_pert}
  \left|\la g, \cLpert^* g \raa\right| \leq X^T T(\xi,\eta) X + \frac{a(\xi)}{4\eta} \| (\nabla_p+\nabla_q) \nabla_p g \|_\Ltwo^2,
\end{equation} 
where
\[
T(\xi,\eta) = \begin{pmatrix} 
  T_{pp}(\xi,\eta) \Id_d & T_{qp}(\xi,\eta) \Id_d /2 \\
  T_{qp}(\xi,\eta) \Id_d / 2 & T_{qq}(\xi,\eta) \Id_d
\end{pmatrix},
\]
the elements of the matrix $T$ being
\[
T_{pp}(\xi,\eta) = \frac{1}{K_\kappa}\left(1+\frac{\beta a(\xi)}{m}\right) + \eta a(\xi),
\quad 
T_{qp}(\xi,\eta) = \frac{1}{K_\nu}\left(1+\frac{\beta a(\xi)}{m}\right) + \frac{\beta a(\xi)}{m K_\kappa} + 2\eta a(\xi),
\quad 
T_{qq}(\xi,\eta) = \frac{\beta a(\xi)}{m K_\nu} + \eta a(\xi).
\]

Finally, in view of the Poincar\'e inequality~\eqref{eq:Poincare_total}, the norm $\| \cdot \|_{a(\xi)}$ can be controlled by $|X|$ as
\[
\| g \|_{a(\xi)}^2 \leq X^T P(\xi) X, \qquad P(\xi) = \begin{pmatrix} \dps a(\xi) + \frac{1}{K_\kappa^2} & a(\xi) \\ a(\xi) & \dps a(\xi) + \frac{1}{K_\nu^2}\end{pmatrix}.
\]
The comparison between~\eqref{eq:control_H1_std} and~\eqref{eq:control_H1_pert} suggests choosing $\eta = \beta|\tau|/(4\xi)$. The constant $\lambda(\xi,\tau)$ can therefore be chosen as the largest real number such that 
\begin{equation}
  \label{eq:condition_matrix_H1}
  S(\xi) - |\tau| T\left(\xi,\frac{\beta|\tau|}{4\xi}\right) \geq \lambda(\xi,\tau)P(\xi).
\end{equation}

Let us now make this condition more explicit by distinguishing the two limiting regimes $\xi \to 0$ and $\xi \to +\infty$.
\begin{enumerate}[(i)]
\item When $\xi \to 0$, the condition that $S$ should be positive definite requires in particular that $S_{pp}(\xi)>0$. This means that $a(\xi)$ should be sufficiently small, in fact at most of order~$\xi$. We therefore consider $a(\xi) = \overline{a}_0 \xi + \mathrm{O}(\xi^2)$. The condition~\eqref{eq:condition_matrix_H1} then reduces to
\[
\begin{aligned}
& \begin{pmatrix}
\dps \xi\left(\frac{1}{\beta} - \overline{a}_0\| \nabla^2 U \|_{L^\infty}\right) - |\tau| \left(\frac{1}{K_\kappa}+\frac{\beta\overline{a}_0|\tau|}{4}\right) & \dps -\frac{\overline{a}_0\xi}{2}\left(\frac1m + \| \nabla^2 U \|_{L^\infty}\right) - \frac{|\tau|}{2K_\nu} - \frac{\beta\overline{a}_0|\tau|^2}{4} \\[10pt]
\dps -\frac{\overline{a}_0\xi}{2}\left(\frac1m + \| \nabla^2 U \|_{L^\infty}\right) - \frac{|\tau|}{2K_\nu} - \frac{\beta\overline{a}_0|\tau|^2}{4} & \dps \frac{\overline{a}_0 \xi}{m} - \frac{\beta\overline{a}_0|\tau|^2}{4}\\
\end{pmatrix} + \mathrm{O}\left(\xi^2, \xi|\tau|\right) \\
& \qquad \qquad \geq  \lambda(\xi,\tau)\left[\begin{pmatrix} \dps \frac{1}{K_\kappa^2} & 0 \\ 0 & \dps \frac{1}{K_\nu^2}\end{pmatrix} + \mathrm{O}(\xi)\right].
\end{aligned}
\]  
It is then clear (by considering for instance the first element on the matrix on the left-hand side) that $\tau$ can be chosen to be at most of order~$\xi$, which we write as $|\tau| = \delta \xi$. The above inequality then reduces to
\[
\begin{pmatrix}
\dps \frac{1}{\beta} - \overline{a}_0\| \nabla^2 U \|_{L^\infty} - \frac{\delta}{K_\kappa} & \dps -\frac{\overline{a}_0}{2}\left(\frac1m + \| \nabla^2 U \|_{L^\infty}\right) - \frac{\delta}{2K_\nu}\\
\dps -\frac{\overline{a}_0}{2}\left(\frac1m + \| \nabla^2 U \|_{L^\infty}\right) - \frac{\delta}{2K_\nu} & \dps \frac{\overline{a}_0}{m}  \\
\end{pmatrix} + \mathrm{O}(\xi) \geq \frac{\lambda(\xi,\delta \xi)}{\xi} \left[\begin{pmatrix} \dps \frac{1}{K_\kappa^2} & 0 \\ 0 & \dps \frac{1}{K_\nu^2}\end{pmatrix} + \mathrm{O}(\xi)\right].
\]
We next choose $\delta>0$ sufficiently small and then $\overline{a}_0$ sufficiently small in order to satisfy the above inequality. In any case, $\lambda(\xi,\tau)$ is of order~$\xi$. Moreover, from the above matrix expression, it is clear that $[\lambda(\xi,\delta\xi) - \lambda(\xi,0)]/\xi = \mathrm{O}(\delta)$.
\item When $\xi \to +\infty$, the limitation on $a(\xi)$ arises from the fact that the determinant of $S(\xi)$ should be positive, which requires $\xi a(\xi)$ to be bounded. We therefore consider $a(\xi) = \overline{a}_\infty/\xi + \mathrm{O}(\xi^{-2})$.The condition~\eqref{eq:condition_matrix_H1} then reduces to
\[
\begin{pmatrix}
\dps \frac{\xi}{\beta} & \dps -\frac{\overline{a}_\infty}{2m} \\[10pt]
\dps -\frac{\overline{a}_\infty}{2m} & \dps \frac{\overline{a}_\infty}{m\xi} \\
\end{pmatrix} -|\tau| \begin{pmatrix} \dps \frac{1}{K_\kappa} & \dps \frac{1}{2K_\nu} \\[10pt] \dps \frac{1}{2K_\nu} & \dps 0\end{pmatrix} + \mathrm{O}\left(\frac{1}{\xi^2}, \frac{|\tau|}{\xi}\right)
\geq  \lambda(\xi,\tau)\left[\begin{pmatrix} \dps \frac{1}{K_\kappa^2} & 0 \\ 0 & \dps \frac{1}{K_\nu^2}\end{pmatrix} + \mathrm{O}\left(\frac{1}{\xi}\right)\right].
\]  
In view of~\eqref{eq:smallest_eig}, the smallest eigenvalue of the first matrix on the left-hand side is of order~$1/\xi$, so that $\tau$ can be at most of order~$1/\xi$, which we write as $\tau = \delta/\xi$. The above inequality then shows that $\lambda(\xi,\delta/\xi)$ is of order~$1/\xi$, and that $\xi [\lambda(\xi,\delta\xi) - \lambda(\xi,0)] = \mathrm{O}(\delta)$.
\end{enumerate}

\subsection{Proof of Theorem~\ref{thm:hypo_L2_like}}
\label{sec:proof_hypo_L2_like}

The key estimate for proving the result is the following hypoelliptic regularization result.

\begin{proposition}
  \label{prop:hypoelliptic_reg}
  There exist $K,\delta_* > 0$ such that, for any $\xi >0$ and $\tau \in [0,\delta_*\min(\xi,1)]$, 
  \[
  \forall g \in L^2(\mu), \quad \forall 0 < t \leq 1, \qquad \left\| \nabla_p \rme^{t \cLxt^*} g \right\|_{L^2(\mu)} +  \left\| \nabla_q \rme^{t \cLxt^*}g \right\|_{L^2(\mu)} \leq \frac{K \max\left(\xi,\xi^{-1}\right)}{t^{3/2}}  \| g \|_{L^2(\mu)}.
  \]
\end{proposition}

Note that the hypoelliptic regularization is possible for values of $\tau$ of order~1 when $\xi$ is large, and it is therefore not this step which limits the range of admissible forcings in Theorem~\ref{thm:hypo_L2_like}.

As a corollary of Proposition~\ref{prop:hypoelliptic_reg}, there exists a constant $\widetilde{K}>0$ such that $\| f(t) \|^2_{a(\xi)} \leq \widetilde{K} t^{-3/2} \| f \|_\Ltwo$ for $0 < t \leq 1$ when the product $a(\xi) \max(\xi,\xi^{-1})$ is bounded, which is the case for the function~$a(\xi)$ considered in Theorem~\ref{thm:hypo_H1_like}. Combining this inequality and Theorem~\ref{thm:hypo_H1_like} with $t_0 = 1$, for instance, we can conclude that, for $t \geq 1$, and any $\xi \in (0,+\infty)$ and $\tau \in [-\delta \mx,\delta \mx]$
\[
\| f(t) \|^2_{L^2(\mu)} \leq \| f(t) \|^2_{a(\xi)} \leq \rme^{-2\overline{\lambda}_\delta \min(\xi,\xi^{-1})(t-t_0)} \| f(t_0) \|^2_{a(\xi)} 
\leq \widetilde{K} \, \rme^{-2\overline{\lambda}_\delta \min(\xi,\xi^{-1})(t-t_0)} \| f(0) \|^2_{L^2(\mu)},
\]
which gives the claimed exponential decay in $L^2(\mu)$. 

\begin{proof}[Proof of Proposition~\ref{prop:hypoelliptic_reg}]
  We denote by $g(t) = \rme^{t \cLxt^*}g$, and introduce 
  \[
  N_g(t) = \frac12\left( \| g(t) \|^2_\Ltwo + A(\xi)t \| \nabla_p g(t) \|^2_\Ltwo+ 2B(\xi)t^2 \la \nabla_p g(t), \nabla_q g(t) \ra + C(\xi)t^3 \| \nabla_q g(t) \|^2_\Ltwo\right),
  \]
  where $A(\xi),C(\xi)$ are positive and $A(\xi)C(\xi)-B(\xi)^2>0$. The result follows provided $N_g(t)$ has a controlled increase (in the sense that $N_g(1)$ can be controlled by $N_g(0)$ up to a multiplicative factor) and $A(\xi)C(\xi)-B(\xi)^2$ is of order $\min(\xi,\xi^{-1})$. 
  
  By computations similar to the ones leading to~\eqref{eq:control_H1_std} (see for instance the proof of~\cite[Theorem~2.18]{LS16}), the following inequality holds when $B(\xi) > 3 m C(\xi)/2$:
  \[
  \begin{aligned}
    \frac{dN_g(t)}{dt} & \leq -\left(\frac{\xi}{\beta} + A(\xi)\left(\frac{\xi t}{m}-\frac12\right)-\|\nabla^2 U\|_{L^\infty}B(\xi)t^2\right) \left\|\nabla_p g(t)\right\|_\Ltwo^2 - \left(\frac{B(\xi)}{m}-\frac{3C(\xi)}{2}\right)t^2 \left\|\nabla_q g(t)\right\|_\Ltwo^2 \\
    & \quad + t \left( 2B(\xi) + \frac{A(\xi)}{m} + \frac{\xi B(\xi)t}{m} + C(\xi)\|\nabla^2 U\|_{L^\infty}t^2 \right) \left\|\nabla_p g(t)\right\|_\Ltwo \left\|\nabla_q g(t)\right\|_\Ltwo\\
    & \quad - \frac{\xi t}{\beta} \left(A(\xi) \left\|\nabla_p^2 g(t)\right\|_\Ltwo^2 - 2B(\xi)t \| \nabla_p^2 g(t) \|_\Ltwo \| \nabla_{qp}^2 g(t) \|_\Ltwo + C(\xi)t^2 \left\|\nabla_{qp}^2 g(t)\right\|_\Ltwo^2 \right) \\
    & \quad + \tau \widetilde{N}_g(t), 
  \end{aligned}
  \]
  with
  \[
  \begin{aligned}
  \widetilde{N}_g(t) & = \la g(t),\cLpert^* g(t)\ra + A(\xi)t \la \nabla_p g(t),\nabla_p \cLpert^* g(t)\ra + C(\xi)t^3 \la \nabla_q g(t),\nabla_q \cLpert^* g(t)\ra \\
& \quad + B(\xi)t^2 \left( \la \nabla_q g(t),\nabla_p \cLpert^* g(t)\ra + \la \nabla_q \cLpert^* g(t),\nabla_p g(t)\ra \right).
  \end{aligned}
  \]
Note that, using $[\partial_{p_i}, \cLpert^*] = \beta F_i/m$,
\[
\begin{aligned}
  \left|\widetilde{N}_g(t)\right| & \leq \| \nabla_p g(t) \|_\Ltwo \| g(t) \|_\Ltwo +  A(\xi)t \| \nabla_p g(t) \|_\Ltwo\left(\frac{\beta}{m}\| g(t) \|_\Ltwo + \| \nabla_p^2 g(t) \|_\Ltwo\right) \\
& \quad + C(\xi)t^3 \| \nabla_q g(t) \|_\Ltwo \| \nabla_{qp}^2 g(t) \|_\Ltwo \\
& \quad + B(\xi)t^2 \left( \| \nabla_p g(t) \|_\Ltwo \| \nabla_{qp}^2 g(t) \|_\Ltwo + \| \nabla_q g(t) \|_\Ltwo \left[ \frac{\beta}{m}\| g(t) \|_\Ltwo + \| \nabla_{p}^2 g(t) \|_\Ltwo\right] \right).
\end{aligned}
\]
We next recall that $\|g(t)\|_\Ltwo^2 \leq 2 N_g(t)$ and use Cauchy--Schwarz inequalities\footnote{Although sharper results may be obtained with Young inequalities, the final scaling of admissible values of $\tau$ in terms of~$\xi$ is unaffected.} such as 
\[
t^2 \| \nabla_p g(t) \|_\Ltwo \| \nabla_{qp}^2 g(t) \|_\Ltwo \leq \frac{t}{2} \Big(\| \nabla_p g(t) \|_\Ltwo^2 + t^2 \| \nabla_{qp}^2 g(t) \|_\Ltwo\Big).
\]
This allows to bound $\widetilde{N}_g(t)$ as follows:
\[
\begin{aligned}
  \left|\widetilde{N}_g(t)\right| & \leq \left(1 + \frac{\beta t(A(\xi)+B(\xi))}{m} \right) N_g(t) \\
& \quad +  \frac12 \left(1  + A(\xi)t \left[1+\frac{\beta}{m}\right] + B(\xi)t \right) \| \nabla_p g(t) \|_\Ltwo^2 + \frac12 \left(C(\xi)+B(\xi)\left[1+\dfrac\beta m\right]\right) t^3 \| \nabla_q g(t) \|_\Ltwo^2 \\
& \quad + \frac12 \left( A(\xi) + B(\xi) \right)t\| \nabla_p^2 g(t) \|_\Ltwo^2 + \frac12 \left( B(\xi) + C(\xi) \right)t^3 \| \nabla_{qp}^2 g(t) \|_\Ltwo^2. 
\end{aligned}
\]
By combining the latter inequality and the above bound on $dN_g/dt$, it follows that, for any $t \in [0,1]$,
\[
\begin{aligned}
  \frac{dN_g(t)}{dt} & \leq \left(1 + \frac{\beta (A(\xi) + B(\xi))}{m}\right) |\tau| N_g(t) \\
  & \quad -\left(\frac{\xi}{\beta} - \frac{A(\xi)}{2} - \|\nabla^2 U\|_{L^\infty}B(\xi) - \frac{|\tau|}{2} \left[1+ A(\xi) \left(1+\frac{\beta}{m}\right) + B(\xi)\right]\right) \left\|\nabla_p g(t)\right\|_\Ltwo^2 \\
& \quad - t^2 \left(\frac{B(\xi)}{m}-\frac{3C(\xi)}{2}-\frac{|\tau|}{2}\left[ C(\xi) + B(\xi)\left(1+\dfrac\beta m\right)\right]\right) \left\|\nabla_q g(t)\right\|_\Ltwo^2 \\
  & \quad + t \left( 2B(\xi) + \frac{A(\xi)}{m} + \frac{\xi B(\xi)}{m} + C(\xi)\|\nabla^2 U\|_{L^\infty} \right) \left\|\nabla_p g(t)\right\|_\Ltwo \left\|\nabla_q g(t)\right\|_\Ltwo\\
  & \quad - \dfrac{\xi t}{\beta} \left(A(\xi) \left\|\nabla_p^2 g(t)\right\|_\Ltwo^2 - 2B(\xi)t \| \nabla_p^2 g(t) \|_\Ltwo \| \nabla_{qp}^2 g(t) \|_\Ltwo + C(\xi)t^2 \left\|\nabla_{qp}^2 g(t)\right\|_\Ltwo^2 \right) \\
& \quad + \frac{|\tau|\left(A(\xi) + B(\xi) \right)t}{2}\| \nabla_p^2 g(t) \|_\Ltwo^2 + \frac{|\tau|\left( B(\xi) + C(\xi)\right)t^3}{2} \| \nabla_{qp}^2 g(t) \|_\Ltwo^2. 
\end{aligned}
\]

The discussion at this stage follows the same strategy as the end of Section~\ref{sec:proof_H1_like}, by rewriting the sum of the second to the fourth lines in matrix form as
\[
-X^T S(\xi,\tau) X, \qquad X = \begin{pmatrix} \left\|\nabla_p g(t)\right\|_\Ltwo \\ \left\|\nabla_{q} g(t)\right\|_\Ltwo \end{pmatrix}, 
\]
and the sum of the last two lines in matrix form as 
\[
-Y^T T(\xi,\tau) Y, \qquad Y = \begin{pmatrix} \left\|\nabla_p^2 g(t)\right\|_\Ltwo \\ \left\|\nabla_{qp}^2 g(t)\right\|_\Ltwo \end{pmatrix}. 
\]
We then distinguish the cases $\xi \to 0$ and $\xi \to +\infty$, and look for conditions ensuring that the matrices $S(\xi,\tau),T(\xi,\tau)$ are nonnegative. It is easily seen that the requirements translate into $A(\xi) = \overline{A} \mx$, $B(\xi) = \overline{B} \mx$, $C(\xi) = \overline{C} \mx$ for positive parameters $\overline{A}, \overline{B}, \overline{C}$ sufficiently small and such that $\overline{A}\overline{C}-\overline{B}^2 > 0$; as well as $\tau \in [-\delta_*\min(\xi,1),\delta_*\min(\xi,1)]$ for $\delta_* > 0$ sufficiently small; and the further conditions that $\overline{C}$ and $\delta_*$ are sufficiently small compared to $\overline{B}$. There exists therefore $R > 0$ (independent of~$\tau$ and~$\xi$) such that
\[
\frac{dN_g(t)}{dt} \leq R |\tau| N_g(t),
\]
from which we deduce that $N_g(1) \leq \rme^{R |\tau|} N_g(0) = 2 \rme^{R |\tau|} \| g \|_\Ltwo^2$. The latter inequality allows to conclude.
\end{proof}

\subsection{Proof of Theorem~\ref{thm:L2_DMS}}
\label{sec:proof_L2_DMS} 

We closely follow the proof of~\cite{DMS09,DMS15}, however specializing the operators under consideration to the case of Langevin dynamics, and considering all operators on~$L^2(\mu)$. This allows to simplify several arguments in the proof.

Introduce the projection $\Pi : L^2(\mu) \to L^2(\nu)$ defined as
\[
(\Pi g)(q) = \langle g(q,\cdot), \mathbf{1}\rangle_{L^2(\kappa)} = \int_{\R^d} g(q,p)\, \kappa(dp).
\]
A simple computation shows that $\Pi \cLham \Pi = 0$ since $(\cLham \Pi g)(q,p) =  p^T \nabla_q (\Pi g)(q)/m$ and $\Pi p = 0$. Consider also the following functional
\[
\mathscr{E}(g) = \frac12 \| g \|_{L^2(\mu)}^2 + a(\xi) \la Ag, g \ra,
\]
for some parameter $a(\xi) \in (0,1)$ to be determined later on in terms of~$\xi$, and with
\[
A = -\left(1 - \Pi \cLham^2 \Pi \right)^{-1} \Pi \cLham.
\]
The latter operator can in fact be made somewhat more explicit, by computing the action of $\cLham^2 \Pi$:
\begin{equation}
  \label{eq:action_Ltwo}
  \cLham^2 \Pi \varphi = \frac{1}{m^2} p^T (\nabla_q^2 \Pi \varphi) p - \frac{1}{m} \nabla U \cdot \nabla_q \Pi \varphi.
\end{equation}
A simple computation then shows that $\Pi \cLham^2 \Pi$ is, up to a multiplicative factor $1/m$, the generator of the overdamped Langevin process~\eqref{eq:overdamped}:
\[
\Pi \cLham^2 \Pi \varphi = \frac{1}{m} \cLovd \Pi \varphi, \qquad \cLovd = - \nabla U \cdot \nabla_q + \frac1\beta \Delta_q\varphi. 
\]
The following result gathers some properties of the operator~$A$ (see \cite[Lemma~1]{DMS15}).

\begin{lemma}
\label{lem:estimates_A}
It holds $\Pi A = A$. Moreover, for any function $g \in L^2(\mu)$,
\[
\| A g \|_{L^2(\mu)} \leq \frac12 \| (1-\Pi)g \|_{L^2(\mu)}, 
\qquad 
\| \cLham A g \|_{L^2(\mu)} \leq \| (1-\Pi)g \|_{L^2(\mu)}. 
\]
\end{lemma}

\begin{proof}
  Consider $g \in \Ltwo$ and $u = Ag$. Then, $(1-\Pi \cLham^2 \Pi)u = -\Pi \cLham g$, so that, taking the scalar product with respect to~$u$, 
\[
\begin{aligned}
\|u\|_\Ltwo^2 + \|\cLham \Pi u\|_\Ltwo^2 & = \langle \cLham \Pi u, (1-\Pi)g \rangle_\Ltwo \\
& \leq \|\cLham \Pi u\|_\Ltwo \|(1-\Pi)g\|_\Ltwo \leq \frac14 \|(1-\Pi)g\|_\Ltwo^2 + \|\cLham \Pi u\|_\Ltwo^2,
\end{aligned}
\]
which gives the claimed result. 
\end{proof}

Denoting by $\sH(t) = \mathscr{E}(f(t))$ with $f(t)$ defined in~\eqref{eq:f(t)}, a simple computation shows that
\[
\begin{aligned}
\sH'(t) & = \langle f(t) , \cLxt f(t) \rangle_{L^2(\mu)} + a(\xi) \langle A\cLxt^* f(t) , f(t) \rangle_{L^2(\mu)} + a(\xi) \langle \cLxt A f(t) , f(t) \rangle_{L^2(\mu)} \\
& = \xi \langle f(t) , \cLFD f(t) \rangle_{L^2(\mu)} + \tau \langle f(t) , \cLpert f(t) \rangle_{L^2(\mu)} \\
& \ \ - a(\xi) \langle A \cLham f(t) , f(t) \rangle_{L^2(\mu)} + \xi a(\xi) \langle A \cLFD f(t) , f(t) \rangle_{L^2(\mu)} + \tau a(\xi) \langle A \cLpert^* f(t) , f(t) \rangle_{L^2(\mu)} \\
& \ \ + a(\xi) \langle \cLham A f(t) , f(t) \rangle_{L^2(\mu)},
\end{aligned}
\]
where we used in the last line that $\cLFD A = \cLFD \Pi A = 0$ and $\cLpert A = \cLpert \Pi A = 0$. Therefore,
\begin{equation}
\label{eq:derivee_sH}
\begin{aligned}
  \sH'(t) & \leq \xi \langle f(t) , \cLFD f(t) \rangle_{L^2(\mu)} + \tau \langle f(t) , \cLpert f(t) \rangle_{L^2(\mu)} \\
  & \ \ - a(\xi) \langle A \cLham f(t) , f(t) \rangle_{L^2(\mu)} + \xi a(\xi) \langle A \cLFD f(t) , f(t) \rangle_{L^2(\mu)} \\
& \ \ + |\tau| a(\xi) \| A \cLpert^* \| \| f(t) \|_{L^2(\mu)}\| \Pi f(t) \|_{L^2(\mu)} + a(\xi) \| (1-\Pi)f(t)\|_{L^2(\mu)}^2 ,
\end{aligned}
\end{equation}
where we used Lemma~\ref{lem:estimates_A} to write $\cLham A = \cLham \Pi A = (1-\Pi)\cLham A$, so that 
\[
\left|\langle \cLham A h , h \rangle_{L^2(\mu)}\right| = \left|\langle \cLham A h , (1-\Pi) h \rangle_{L^2(\mu)}\right| \leq \| (1-\Pi)h\|_{L^2(\mu)}^2.
\]
Let us now consider successively the various terms in~\eqref{eq:derivee_sH}:
\begin{itemize} 
\item The sum first two terms on the right-hand side of~\eqref{eq:derivee_sH} can be bounded as
\[
\begin{aligned}
\xi \langle f(t) , \cLFD f(t) \rangle_{L^2(\mu)} + \tau \langle f(t), \cLpert f(t) \rangle_{L^2(\mu)} & \leq -\frac{\xi}{\beta} \| \nabla_p f(t) \|_{L^2(\mu)}^2 + |\tau| \|f(t)\|_\Ltwo \| \nabla_p f(t)\|_\Ltwo \\
& \leq -\left(\frac{\xi}{\beta}-\frac{|\tau|}{2\eta}\right) \| \nabla_p f(t) \|_{L^2(\mu)}^2 + \frac{\eta|\tau|}{2} \|f(t)\|_\Ltwo^2, \\
& \leq -K_\kappa^2\left(\frac{\xi}{\beta}-\frac{|\tau|}{2\eta}\right) \| (1-\Pi)f(t) \|_{L^2(\mu)}^2 + \frac{\eta|\tau|}{2} \|f(t)\|_\Ltwo^2,
\end{aligned}
\]
where $\eta>\beta|\tau|/(2\xi)$, and where we used $\nabla_p \Pi = 0$ and a Poincar\'e inequality to obtain the last inequality (since $(1-\Pi)g(t,q,\cdot) \in L^2_0(\kappa)$ for almost all~$q$). 
\item We rewrite the first term of the second line of~\eqref{eq:derivee_sH} as 
\begin{equation}
\label{eq:ALham}
\langle A \cLham f(t) , f(t) \rangle_{L^2(\mu)} = \langle A \cLham \Pi f(t) , f(t) \rangle_{L^2(\mu)} + \langle A \cLham (1-\Pi) f(t) , f(t) \rangle_{L^2(\mu)}.
\end{equation}
We start with the first term on the right-hand side of the above equality. Denoting by $B = \cLham \Pi$, it holds $(B h)(q,p) = p^T \nabla_q (\Pi h)(q)/m$. When $h \in L^2_0(\mu)$, and since $\nu$ satisfies the Poincar\'e inequality~\eqref{eq:Poincare_nu}, it holds
\[
\| Bh \|_{L^2(\mu)} = \sqrt{\frac{d}{m\beta}} \| \nabla_q (\Pi h) \|_{L^2(\nu)} \geq K_\nu \sqrt{\frac{d}{m\beta}} \| \Pi h \|_{L^2(\nu)}.
\]
This can be rephrased as
\[
B^* B \geq \frac{d K_\nu^2}{m\beta} \Pi
\]
in the sense of symmetric operators on $L^2_0(\mu)$. Since $A\cLham \Pi = (1+B^*B)^{-1} B^* B = 1 - (1+B^*B)^{-1}$, we can conclude that
\[
- \langle A \cLham \Pi f(t) , f(t) \rangle_{L^2(\mu)} \leq -\frac{d K_\nu^2/(m\beta)}{1 + d K_\nu^2/(m\beta)} \| \Pi f(t) \|_{L^2(\mu)}^2.
\]
For the second term on the right-hand side of~\eqref{eq:ALham}, we write (using $\Pi A = A$)
\[
\langle A \cLham (1-\Pi) f(t) , f(t) \rangle_{L^2(\mu)} = -\langle (1-\Pi)f(t) , \cLham A^* \Pi f(t) \rangle_{L^2(\mu)}. 
\]
By Lemma~\ref{lem:Lham_A*} below, the operator $\cLham A^*$ is bounded, so that the absolute value of the right-hand side of the above equality is bounded by $\| \cLham A^* \| \| (1-\Pi)f(t)\|_{L^2(\mu)} \| \Pi f(t)\|_{L^2(\mu)}$.
\item In order to treat the second term in the second line of~\eqref{eq:derivee_sH}, we compute the action of $A \cLFD$. Now, $A\cLFD = -(1-\Pi \cLham^2\Pi)^{-1} \Pi \cLham \cLFD = -(1-\Pi \cLham^2\Pi)^{-1} \Pi [\cLham,\cLFD]$ since $\Pi \cLFD = 0$. In order to evaluate the commutator, we compute
\[
[\cLham,\cLFD]h = \frac{1}{m} \nabla U \cdot \nabla_p h + \frac{p}{m^2}\cdot \nabla_q h - \frac{2}{m\beta}\sum_{i=1}^d \partial_{q_i p_i}^2 h. 
\]
Upon applying $\Pi$ to the various terms and noting that 
\[
\int_{\R^d} \partial_{q_i p_i}^2 h \, d\kappa = \frac{\beta}{m} \int_{\R^d} p_i \partial_{q_i}h \, d\kappa,
\]
it follows that $\Pi[\cLham,\cLFD] = -\Pi\cLham/m$. Therefore, $A\cLFD = -A/m$ is bounded by Lemma~\ref{lem:estimates_A}. More precisely, 
\[
\left| \langle A \cLFD f(t) , f(t) \rangle_{L^2(\mu)} \right| \leq \frac{1}{2m} \| (1-\Pi)f(t)\|_{L^2(\mu)} \| \Pi f(t)\|_{L^2(\mu)}.
\]
\item Finally, $A \cLpert^*$ is bounded by Lemma~\ref{lem:Lham_A*} below, with $\|A \cLpert^*\| \leq \sqrt{\beta/(4m)}$.
\end{itemize}

Gathering all estimates, we obtain 
\[
\sH'(t) \leq -X(t)^T \Big( S(\xi) - |\tau|T(\xi) \Big) X(t),
\]
with
\[
X = \begin{pmatrix} \| \Pi f(t)\|_{L^2(\mu)} \\ \| (1-\Pi) f(t)\|_{L^2(\mu)} \end{pmatrix}, 
\qquad 
S(\xi) = \begin{pmatrix} S_{--}(\xi) & S_{-+}(\xi)/2 \\ S_{-+}(\xi)/2 & S_{++}(\xi) \end{pmatrix},
\qquad
T(\xi) = \begin{pmatrix} T_{--}(\xi) & T_{-+}(\xi)/2 \\ T_{-+}(\xi)/2 & T_{++}(\xi) \end{pmatrix},
\]
where
\[
\begin{cases}
S_{--}(\xi) \dps = a(\xi) \frac{d K_\nu^2/(m\beta)}{1 + d K_\nu^2/(m\beta)},\\[10pt]
S_{-+}(\xi) \dps = - a(\xi) \left( \| \cLham A^* \| + \frac{\xi}{2m} \right),\\[10pt] 
S_{++}(\xi) \dps = \frac{\xi K^2_\kappa}{\beta} - a(\xi),\\ 
\end{cases} 
\qquad 
\begin{cases}
T_{--}(\xi) \dps = \frac12 \left( a(\xi)\sqrt{\frac{\beta}{m}} + \eta\right),\\[10pt] 
T_{-+}(\xi) \dps = \frac{a(\xi)}{2}\sqrt{\frac{\beta}{m}},\\[10pt] 
T_{++}(\xi) \dps = \frac12 \left( \eta+\frac{K_\kappa^2}{\eta}\right).\\ 
\end{cases}. 
\]
Therefore, since $2\sH(t) \leq (1+a(\xi))\|f(t)\|_{L^2(\mu)}^2$ by Lemma~\ref{lem:estimates_A},
\begin{equation}
  \label{eq:key_estimate_DMS}
  \sH'(t) \leq -\Lambda_-\Big(S(\xi)-|\tau|T(\xi)\Big) \|f(t)\|^2_{L^2(\mu)} \leq - \frac{2\Lambda_-\Big(S(\xi)-|\tau|T(\xi)\Big)}{1+a(\xi)}\sH(t),
\end{equation}
where $\Lambda_-$ is defined in~\eqref{eq:smallest_eig}. We next follow a discussion similar to the one performed at the end of Section~\ref{sec:proof_H1_like}. When $\tau = 0$, the requirement that $S(\xi)$ is positive definite shows that $a(\xi)$ should be chosen of order~$\mx$. Next, in order for $S(\xi)-|\tau|T(\xi)$ to be positive definite for $\tau \neq 0$, we see that $\tau$ should be of order~$\xi$ when $\xi \to 0$ (choosing $\eta = \beta\tau/(4\xi)$ for instance), but can be taken to be of order~1 as $\xi \to +\infty$ by setting $\eta = 1/\xi$. Indeed, this choice leads to 
\[
S(\xi)-|\tau|T(\xi) \simeq \begin{pmatrix} \dps \frac{\overline{a}_\infty}{\xi} \frac{d K_\nu^2/(m\beta)}{1 + d K_\nu^2/(m\beta)}& \dps -\frac{\overline{a}_\infty}{4m}\\ \dps -\frac{\overline{a}_\infty}{4m} & \dps \frac{\xi K^2_\kappa}{\beta}\end{pmatrix} - \frac{|\tau|}{2} \begin{pmatrix} \dps \frac{1 + \overline{a}_\infty \sqrt{\beta/m}}{\xi} & 0 \\ 0 & \dps \frac{\xi K_\kappa^2}{2}\end{pmatrix},
\]
where $\overline{a}_\infty$, the limit of $\xi a(\xi)$ as $\xi \to +\infty$, should be sufficiently small. In conclusion, it is possible to set $\tau = \delta \min(\xi,1)$ for $|\delta|$ sufficiently small, and, for such a choice, there exists $\overline{\lambda}_\delta >0$ for which, for any $\xi \in (0,+\infty)$,
\[
\sH'(t) \leq -\overline{\lambda}_\delta \mx \| f(t) \|_{L^2(\mu)}^2.
\]
Moreover, $\overline{\lambda}_\delta = \overline{\lambda}_0 + \mathrm{O}(\delta)$. 

The final result is obtained by noting that, in view of Lemma~\ref{lem:estimates_A}, and considering functions $a$ with values in a compact subset of~$(0,1)$,
\[
\| f(t) \|_{L^2(\mu)}^2 \leq \frac{2}{1-a(\xi)} \sH(t), \qquad \sH(0) \leq \frac{1+a(\xi)}{2} \| f(0) \|^2_\Ltwo.
\]
This shows that the constant $C$ in~\eqref{eq:cv_L2_like_DMS} can be chosen as (by restricting $a(\xi)$ to remain lower than 1/2 for instance)
\[
C = \sup_{\xi > 0} \sqrt{\frac{1+ a(\xi)}{1 - a(\xi)}}.
\]

\medskip

It remains to prove the following technical result.
\begin{lemma}
  \label{lem:Lham_A*}
  The operators $\cLham A^* = \cLham^2 \Pi (1-\Pi \cLham^2 \Pi)^{-1}$ and $\cLpert A^*$ are bounded. Moreover, $\|\cLpert A^*\| \leq \sqrt{\beta/(4m)}$.  
\end{lemma}

\begin{proof}
In view of~\eqref{eq:action_Ltwo}, the action of $\cLham A^* = \cLham A^* \Pi$ is
\[
\cLham A^* \Pi \varphi(q,p) = \frac{1}{m^2} p^T (\nabla_q^2 \Pi \psi)(q) p - \frac{1}{m} \nabla U(q) \cdot \nabla_q \Pi \psi(q),
\]
with
\[
\psi = \left(1 - \frac1m \cLovd \right)^{-1} \Pi \varphi.
\]
The operator $\cLham A^* \Pi$ is then bounded since $1 - \cLovd/m$ is a bounded operator from $L^2(\nu)$ to $H^2(\nu)$. 

We next consider $\cLpert A^* = \cLpert \cLham \Pi \left(1 - \cLovd/m\right)^{-1} \Pi$. Since
\[
\cLpert \cLham \Pi h = \frac{1}{m} \cLpert\left(p \cdot \nabla_q \Pi h \right) = \frac1m F \cdot \nabla_q \Pi h,
\]
the operator $\cLpert A^*$ is also bounded since $1 - \cLovd/m$ is a bounded operator from $L^2(\nu)$ to $H^1(\nu)$. Moreover, for $\varphi \in \Ltwo$,
\[
\begin{aligned}
\| \cLpert A^* \varphi \|^2 & = \langle \varphi, A \cLpert^* \cLpert A^* \varphi \rangle_\Ltwo \\
& = \frac{1}{m^2} \left\langle \varphi, \Pi(1-\cLovd/m)^{-1} (F\cdot \nabla_q)^* F\cdot \nabla_q (1-\cLovd/m)^{-1} \Pi \varphi \right\rangle_\Ltwo. 
\end{aligned}
\]
Since $0 \leq (F\cdot \nabla_q)^* F\cdot \nabla_q \leq \nabla_q^* \nabla_q = -\beta \cLovd$ in the sense of symmetric operators, 
it follows that, by spectral calculus,
\[
\| \cLpert A^* \|^2 \leq \sup_{x \geq 0} \frac{\beta x}{m^2(1+x/m)^2} = \frac{\beta}{m}\sup_{y \geq 0} \frac{y}{(1+y)^2} = \frac{\beta}{4m},
\]
which gives the claimed bound. 
\end{proof}

\subsection*{Acknowledgements}

We thank Anton Arnold for interesting discussions. The work of G.S. is supported by the Agence Nationale de la Recherche, under grant ANR-14-CE23-0012 (COSMOS) and by the European Research Council under the European Union's Seventh Framework Programme (FP/2007-2013) / ERC Grant Agreement number 614492. The work of S.O. is supported by the Agence Nationale de la Recherche, under grant ANR-15-CE40-0020-01 (LSD). 


\end{document}